\documentclass[journal]{IEEEtran}

\usepackage{amssymb}
\usepackage{amsmath}
\usepackage{cite}
\usepackage{url}
\usepackage{empheq}
\usepackage{xcolor}
\usepackage{graphicx}
\usepackage{subfigure}
\usepackage{enumitem}
\usepackage{fancyhdr}
\usepackage{mdwmath}	
\usepackage{mdwtab}
\usepackage{caption}
\usepackage{amsthm}
\usepackage{algorithm}
\usepackage{algorithmic}

\newtheorem{lemma}{Lemma}
\newtheorem{remark}{Remark}
\newtheorem{theorem}{Theorem}
\newtheorem{corollary}{Corollary}
\newtheorem{assumption}{Assumption}
\newtheorem{definition}{Definition}
\newtheorem{proposition}{Proposition}

\newcommand{\eqr}[1]{(\ref{#1})}
\newcommand{\fref}[1]{Fig.~\ref{#1}}

\begin{document}
\title{Leaky Coaxial Cable based Generalized Pinching-Antenna Systems with Dual-Port Feeding}
\author{Kaidi~Wang,~\IEEEmembership{Member,~IEEE,}
Zhiguo~Ding,~\IEEEmembership{Fellow,~IEEE,}
and Daniel~K.~C.~So,~\IEEEmembership{Senior~Member,~IEEE}
\thanks{Kaidi~Wang and Daniel~K.~C.~So are with the Department of Electrical and Electronic Engineering, the University of Manchester, Manchester, M1 9BB, UK (email: kaidi.wang@ieee.org; d.so@manchester.ac.uk).}
\thanks{Zhiguo~Ding is with the School of Electrical and Electronic Engineering (EEE), Nanyang Technological University, Singapore 639798 (e-mail: zhiguo.ding@ntu.edu.sg).}}
\maketitle
\begin{abstract}
By leveraging the distributed leakage radiation of leaky coaxial cables (LCXs), the concept of pinching antennas can be generalized from the conventional high-frequency waveguide based architectures to cable based structures in lower-frequency scenarios. This paper investigates an LCX based generalized pinching-antenna system with dual-port feeding. By enabling bidirectional excitation along each cable, the proposed design significantly enhances spatial degrees of freedom. A comprehensive channel model is developed to characterize intra-cable attenuation, bidirectional phase progression, slot based radiation, and wireless propagation. Based on this model, both analog and hybrid beamforming frameworks are studied with the objective of maximizing the minimum achievable data rate. For analog transmission, slot activation, port selection, and power allocation are jointly optimized using matching theory, coalitional game, and bisection based power control. For hybrid transmission, zero-forcing (ZF) digital precoding is incorporated to eliminate inter-user interference, thereby simplifying slot activation and enabling closed-form optimal power allocation. Simulation results demonstrate that dual-port feeding provides notable performance gains over single-port LCX systems and fixed-antenna benchmarks, validating the effectiveness of the proposed beamforming and resource allocation designs under various transmit power levels and cable parameters.
\end{abstract}
\section{Introduction}
The rapid evolution of wireless communication networks has driven a continuous increase in demand for flexible and reconfigurable antenna technologies, particularly in complex propagation environments such as urban areas, transportation systems, and industrial deployments. To adapt wireless channels to dynamic environments, a variety of flexible antenna paradigms have been extensively investigated, including reconfigurable intelligent surfaces (RISs), fluid antennas, and movable antennas \cite{wu2019irs, wong2020fluid, zhu2023modeling}. RISs reconfigure wireless channels by adjusting the electromagnetic responses of passive reflecting elements, thereby reshaping the composite propagation paths between transmitters and receivers \cite{wu2021intelligent}. Fluid antennas exploit conductive fluids to dynamically reconfigure antenna positions within a constrained physical space \cite{wu2024fluid}, while movable antennas enable spatial adaptation through mechanical displacement of radiating elements \cite{ma2024movable}. These technologies have demonstrated significant potential in improving link reliability and spectral efficiency through precise channel adaptation. However, such reconfiguration is generally realized through indirect channel manipulation via cascaded transmitter-surface-receiver links or wavelength-scale element displacement, which limits direct control over the physical radiation location and large-scale propagation geometry.

To meet the need for large-scale spatial reconfigurability, the concept of pinching antennas has recently emerged as a promising solution \cite{ding2024pin}. By attaching controllable radiating elements, commonly referred to as pinches, onto a guided-wave structure, the electromagnetic energy propagating along the medium can be intentionally leaked at selected locations \cite{liu2025pinching}. As a result, the radiation position becomes dynamically reconfigurable over distances far exceeding the carrier wavelength, enabling substantial modification of wireless channel conditions \cite{kaidi2025pin, xu2025pin}. Such a mechanism allows pinching antennas to reshape spatial radiation patterns, improve coverage performance, and establish line-of-sight (LoS) links in scenarios with severe blockage \cite{kaidi2025pin4, xie2026pin}. Compared with conventional flexible antenna technologies, pinching antennas provide a fundamentally different reconfiguration paradigm by directly exploiting guided-wave propagation, while retaining attractive features such as low hardware cost, minimal reliance on active radio-frequency (RF) components, and high implementation flexibility \cite{fang2025pin}.
\subsection{Related Works}
Building upon the fundamental concept of pinching antennas, extensive research efforts have been devoted to extending system models and enhancing communication capabilities. From an implementation perspective, a practical pinching-antenna system was proposed by introducing pre-installed pinching antennas at discrete candidate locations, thereby transforming the continuous antenna placement problem into a discrete antenna activation framework \cite{kaidi2025pin2}. Beyond modeling realism, several studies have explored waveguide configuration as an additional design dimension for improving deployment flexibility and transmission performance \cite{ouyang2025uplink, zhong2025pin}. In this context, a segmented waveguide enabled pinching-antenna system was proposed for physically consistent uplink communications by eliminating inter-antenna re-radiation effects, with reduced in-waveguide attenuation and improved maintainability through modular deployment \cite{ouyang2025uplink}. In contrast, a two-dimensional pinching-antenna architecture extended conventional line shaped waveguides into planar configurations, enabling reconfigurable apertures for enhanced spatial beamforming \cite{zhong2025pin}. At the system paradigm level, waveguide division multiple access (WDMA) was proposed to assign each user to a dedicated waveguide, in which waveguides are treated as a new type of radio resource for multi-user communications \cite{zhao2026pass}. Meanwhile, environment division multiple access (EDMA) was introduced by controllably establishing and blocking LoS links with pinching antennas, thereby facilitating spatial user separation and interference suppression beyond conventional beam domain multiplexing \cite{ding2025edma}.

From the signal domain transmission perspective, enhanced transmission architectures have also been investigated. In \cite{zhao2026pin}, three transmission structures for multi-waveguide pinching-antenna systems were proposed to support multi-user communications through joint baseband and pinching beamforming. A fully connected tri-hybrid beamforming architecture was designed in \cite{zhao2025tri}, where a tunable phase-shifter network interconnects all RF chains and waveguides to realize joint digital, analog, and pinching beamforming. In addition, a multi-mode pinching-antenna system exploited a waveguide supporting the simultaneous propagation of multiple guided modes, enabling mode domain multiplexing for multi-user communications \cite{xu2026pin}. Meanwhile, the feeding architecture of pinching-antenna systems has attracted growing attention, as it fundamentally determines the excitation manner and propagation characteristics of the guided waves. Specifically, a wireless-fed pinching-antenna system was proposed to improve deployment flexibility \cite{wijewardhana2025pin}, whereas a center-fed architecture exploited controllable power splitting to enable bidirectional propagation and double the spatial degrees of freedom \cite{gan2026pin}.
\subsection{Motivation and Contributions}
Despite the extensive system-level advancements achieved in pinching-antenna research, existing studies have predominantly focused on dielectric waveguide based implementations operating at high-frequency bands, which limits the applicability of pinching-antenna systems in sub-6 GHz wireless scenarios. Motivated by this limitation, the pinching-antenna paradigm has been extended to more general guiding structures, with leaky coaxial cables (LCXs) identified as a promising physical platform \cite{xu2025generalized}. In \cite{kaidi2025generalized}, an LCX based generalized pinching-antenna framework was developed through slot activation, indicating that the phase variations across activated radiation slots can be exploited to realize analog beamforming. Building on this framework, dual-port feeding, an established technique in conventional LCX systems \cite{hou2017lcx, wu2017lcx}, is introduced in this work to enable bidirectional guided-wave propagation. Specifically, the two feed ports induce distinct phase responses at each radiation slot, thereby enhancing the design flexibility of analog beamforming, while providing two effective transmit dimensions per cable for digital beamforming. The resulting additional degrees of freedom, combined with dynamic slot activation, also lead to distinct channel characteristics and beamforming challenges associated with port-dependent phase responses and reconfigurable radiation patterns. These opportunities and challenges motivate the investigation of the proposed architecture. The main contributions of this paper are summarized as follows.
\begin{itemize}[leftmargin=*]
\item A dual-port fed LCX based generalized pinching-antenna system is proposed, where bidirectional excitation fundamentally changes the guided-wave propagation and radiation characteristics. Key insights into the resulting channel structure and interference behavior are revealed, which serve as the basis for the subsequent design of analog and hybrid transmission schemes. A corresponding channel model is established to capture the composite effects of bidirectional guided transmission and wireless radiation.
\item Based on the developed dual-port LCX channel model, two max-min optimization problems are formulated to maximize the minimum achievable data rate under the analog and hybrid transmission schemes, respectively. For the analog scheme, slot activation, port selection, and power allocation are jointly considered, whereas for the hybrid scheme, slot activation and power allocation are optimized under digital precoding to ensure user fairness.
\item For the analog transmission scheme, an efficient solution framework is developed to solve the resulting mixed-integer optimization problem through problem decomposition. Port selection is formulated as a perfect matching, within which slot activation is updated via a coalitional game, while power allocation is optimized using a bisection based approach, enabling low-complexity analog transmission design under bidirectional LCX excitation.
\item For the hybrid transmission scheme, zero-forcing (ZF) beamforming is incorporated into the analog design to suppress inter-user interference. As a result, all activated slots can be treated as a single coalition in the corresponding coalitional game formulation. The resulting interference-free transmission structure enables a closed-form solution for power allocation, leading to an efficient hybrid beamforming design under bidirectional LCX excitation.
\end{itemize}
Simulation results are presented to validate the proposed dual-port fed LCX based generalized pinching-antenna architecture and the effectiveness of the developed analog and hybrid transmission schemes. The results demonstrate the advantages of dual-port feeding and confirm that the proposed solutions significantly improve the minimum achievable rate compared with conventional single-port LCX configurations and fixed-antenna benchmark systems.
\section{System Model}
Consider an LCX based generalized pinching-antenna system, as shown in \fref{system}, consisting of $K$ parallel cables, each equipped with $M$ radiation slots, serving a total of $N$ users, where $N=2K$\footnote{This assumption is adopted to simplify the system configuration and enable a direct comparison between the analog and hybrid beamforming schemes under a unified setting. The framework can be extended to $N\neq2K$ by utilizing only the required transmit dimensions for $N<2K$, or by introducing additional orthogonal resources or advanced multiple-access schemes such as non-orthogonal multiple access (NOMA) for $N>2K$.}. All cables are installed at a uniform height $d$ and are oriented along the $x$-axis. The sets of cables and users are denoted by $\mathcal{K}=\{1,2,\dots,K\}$ and $\mathcal{N}=\{1,2,\dots,N\}$, respectively, and the slot index set for cable $k$ is given by $\mathcal{M}_k=\{1,2,\dots,M\}$. In the considered system, users are randomly distributed within a rectangular service region of dimensions $D_x\times D_y$, centered at $(0,0,0)$. The location of user $n$ is denoted by $\boldsymbol{\psi}_n=(x_n,y_n,0)$. The $k$-th cable is located at $y=y_k$, and the position of the $m$-th slot on cable $k$ is given by $\boldsymbol{\psi}_{k,m}^\mathrm{slot}=(x_m,y_k,d)$, where adjacent slots are uniformly spaced such that $x_{m+1}-x_m=\Delta d$.

\begin{figure}[!t]
\hspace{2mm}\centering{\includegraphics[width=85mm]{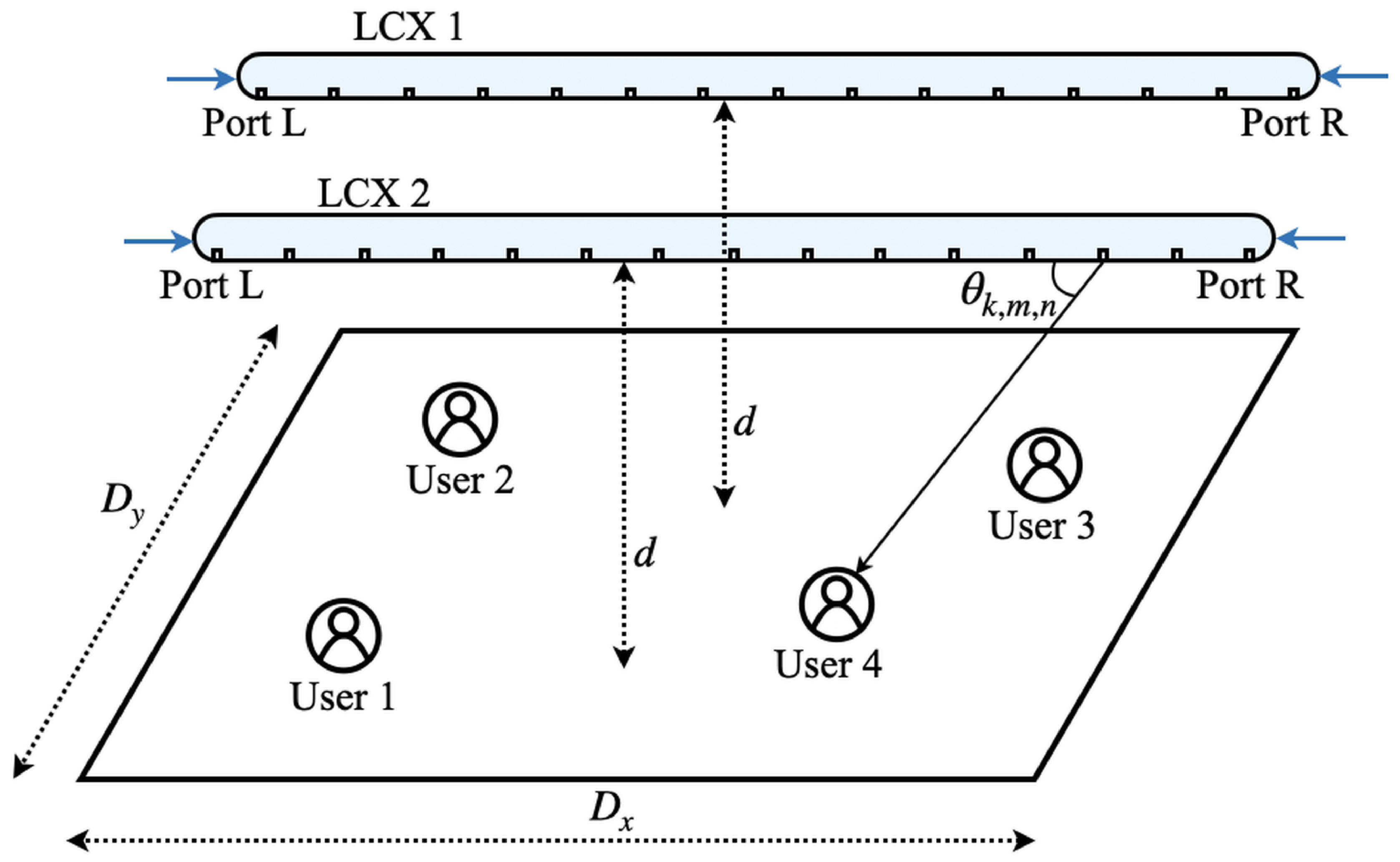}}
\caption{An illustration of the proposed LCX based generalized pinching-antenna system operating in dual-port mode.}\vspace{-4mm}
\label{system}
\end{figure}
\subsection{Channel Model}
As shown in \fref{system}, the dual-port operating mode is considered, in which each cable is excited from both ends, referred to as ports $\mathrm{L}$ and $\mathrm{R}$ \cite{hou2017lcx, zhang2022channel}. Owing to the different propagation distances along the cable, the signals injected from the two ports experience distinct attenuation and phase variations\footnote{To enable dual-port excitation in practice, both ports can be driven by a single base station via external feeder cables. Since these feeder cables are low-loss transmission lines, the associated loss is typically negligible compared with the longitudinal attenuation of the LCX and is therefore omitted from the channel model.} \cite{torrance1996lcx}. For cable $k$, the intra-cable channel between port $\mathrm{L}$, located at the left end of the cable, and the $m$-th slot is expressed as
\begin{equation}
h_{k,m}^\mathrm{L}=10^{-\frac{\kappa_r}{20}(m-1)\Delta d}e^{-j\frac{2\pi}{\lambda}\sqrt{\varepsilon_r}(m-1)\Delta d},
\end{equation}
where $\kappa_r$ is the attenuation constant of the coaxial cable in decibels per meter, $\lambda$ is the free-space wavelength, and $\varepsilon_r$ is the relative permittivity of the dielectric material. Similarly, for port $\mathrm{R}$ located at the right end of the cable, the propagation distance to the $m$-th slot is $(M-m)\Delta d$, and the corresponding intra-cable channel is given by
\begin{equation}
h_{k,m}^\mathrm{R}=10^{-\frac{\kappa_r}{20}(M-m)\Delta d}e^{-j\frac{2\pi}{\lambda}\sqrt{\varepsilon_r}(M-m)\Delta d}.
\end{equation}

Based on the above intra-cable channel model, the following remark provides insight into the attenuation reduction achieved by dual-port feeding.
\begin{remark}
Compared with single-port feeding, dual-port feeding can substantially reduce the intra-cable attenuation. Under single-port operation, the maximum propagation distance is attained at the farthest slot and equals $(M-1)\Delta d$. In contrast, with dual-port feeding, each slot is effectively served by its nearest port, and the maximum propagation distance is reduced to at most $\frac{M-1}{2}\Delta d$. As a result, the maximum intra-cable propagation loss in the dB domain is approximately halved.
\end{remark}

In addition to enhancing the desired signal strength, dual-port feeding can also facilitate multi-user interference mitigation, as summarized in the following remark.
\begin{remark}
When the desired signal for each user is injected through a single selected port, the signals intended for other users from another port typically propagate over a longer distance along the cable before being radiated toward a given user. Consequently, these interfering components experience stronger intra-cable attenuation, which effectively reduces the interference power observed at the users.
\end{remark}

In the considered system, each coaxial cable is installed with the radiation slots oriented downward. As a result, each slot behaves as a magnetic dipole radiator and directs electromagnetic radiation towards the user plane beneath the cable. As reported in \cite{morgan1999lcx}, the downward radiation pattern of each slot follows Lambert's cosine law. When expressed in terms of the elevation angle, this Lambertian characteristic introduces a $\sin(\theta)$ factor, which captures the angular dependence of the radiated energy towards the user plane \cite{yin2024lcx}. Accordingly, the LoS component of the channel between slot $m$ on cable $k$ and user $n$ is given by
\begin{equation}
h_{k,m,n}^\mathrm{LoS}=\eta\frac{e^{-j\frac{2\pi}{\lambda}\left\|\boldsymbol{\psi}_n-\boldsymbol{\psi}_{k,m}^\mathrm{slot}\right\|}}{\big\|\boldsymbol{\psi}_n-\boldsymbol{\psi}_{k,m}^\mathrm{slot}\big\|}\sin\!\left(\theta_{k,m,n}\right),
\end{equation}
where $\eta=\frac{c}{4\pi f_c}$ is the free-space path-loss constant, with $c$ denoting the speed of light and $f_c$ the carrier frequency, $\|\boldsymbol{\psi}_n-\boldsymbol{\psi}_{k,m}^\mathrm{slot}\|$ is the Euclidean distance between slot $m$ on cable $k$ and user $n$, and $\theta_{k,m,n}$ is the elevation angle from slot $m$ on cable $k$ to user $n$. Specifically, due to the vertical separation $d$ between the slot and the user plane, as shown in \fref{system}, the elevation angle satisfies
\begin{equation}
\sin(\theta_{k,m,n})=\frac{d}{\big\|\boldsymbol{\psi}_n-\boldsymbol{\psi}_{k,m}^\mathrm{slot}\big\|}.
\end{equation}
That is, the $\sin(\theta_{k,m,n})$ term characterizes the angular dependence of the slot radiation pattern, where larger elevation angles correspond to stronger radiation toward the user plane, while smaller elevation angles correspond to weaker radiation.

Since the considered LCX based system operates in sub-6 GHz bands and supports more general outdoor deployments, reflected and scattered propagation from surrounding objects may contribute non-negligible non-line-of-sight (NLoS) components. Accordingly, $L$ scattering points are introduced to model such NLoS propagation. The resulting channel between slot $m$ on cable $k$ and user $n$ is expressed as follows:
\begin{equation}
h_{k,m,n}^\mathrm{NLoS}\!=\!\eta\!\sum_{\ell=1}^{L}\!\frac{\delta_\ell e^{-j\frac{2\pi}{\lambda}\left(\left\|\boldsymbol{\psi}_\ell^\mathrm{scat}\!-\boldsymbol{\psi}_{k,m}^\mathrm{slot}\right\|+\left\|\boldsymbol{\psi}_n\!-\boldsymbol{\psi}_\ell^\mathrm{scat}\right\|\right)}}{\big\|\boldsymbol{\psi}_\ell^\mathrm{scat}-\boldsymbol{\psi}_{k,m}^\mathrm{slot}\big\|\left\|\boldsymbol{\psi}_n-\boldsymbol{\psi}_\ell^\mathrm{scat}\right\|}\!\sin\!\left(\theta_{k,m,\ell}\right),
\end{equation}
where $\delta_\ell\sim\mathcal{CN}(0,1)$ is the complex gain associated with the propagation path via scatterer $\ell$, $\boldsymbol{\psi}_\ell^\mathrm{scat}=(x_\ell,y_\ell,z_\ell)$ is the location of scatterer $\ell$, $\|\boldsymbol{\psi}_\ell^\mathrm{scat}-\boldsymbol{\psi}_{k,m}^\mathrm{slot}\|$ and $\|\boldsymbol{\psi}_n-\boldsymbol{\psi}_\ell^\mathrm{scat}\|$ are the distances from slot $m$ on cable $k$ to scatterer $\ell$ and from scatterer $\ell$ to user $n$, respectively, and $\theta_{k,m,\ell}$ is the elevation angle from slot $m$ on cable $k$ to scatterer $\ell$.

The overall channel\footnote{In this work, perfect channel state information (CSI) is assumed to be available at the BS. Since the LCX related parameters are fixed and can be characterized in advance, the instantaneous CSI mainly depends on the user locations and scattering environment.} from feed port $\mathrm{X}\in\{\mathrm{L},\mathrm{R}\}$ to user $n$ via slot $m$ on cable $k$ is given by
\begin{equation}
h_{k,m,n}^\mathrm{X}=h_{k,m}^\mathrm{X}\!\left(h_{k,m,n}^\mathrm{LoS}+h_{k,m,n}^\mathrm{NLoS}\right).
\end{equation}
Accordingly, the channel vector from the feed port of cable $k$ to user $n$ can be presented as follows:
\begin{equation}
\mathbf{h}_{k,n}^\mathrm{X}=\left[h_{k,1,n}^\mathrm{X},h_{k,2,n}^\mathrm{X},\dots,h_{k,M,n}^\mathrm{X}\right]^\mathsf{T}.
\end{equation}
\subsection{Analog Beamforming via Slot Activation}
Since the signals transmitted from different feed ports along the same cable experience distinct phase variations, each slot can radiate signals with identical amplitudes but different phases. By selectively activating or deactivating these slots\footnote{Slot activation and deactivation have been previously considered for LCX based communications in \cite{nagayama2022lcx1, hou2023lcx} and can be implemented through various LCX reconfiguration mechanisms \cite{asplund2016leaky}. Specifically, conductive or absorbing coverings can be applied to locally shield or expose selected radiation slots, while reconfigurable conductive structures can control the radiation state of individual slots by adjusting the shape, size, or position.}, pinching based analog beamforming can thus be realized through binary radiation control. This mechanism follows the concept of pinching beamforming introduced in \cite{wang2025pa, kaidi2025pin3}, in which spatial beam shaping is achieved without explicit phase shifters. The slot activation vector of cable $k$ is defined as
\begin{equation}
\boldsymbol{\alpha}_k=[\alpha_{k,1}, \alpha_{k,2}, \dots, \alpha_{k,M}]^{\mathsf{T}},
\end{equation}
where $\alpha_{k,m}$ is the activation indicator of slot $m$ on cable $k$, with $\alpha_{k,m}=1$ indicating that the slot is activated and $\alpha_{k,m}=0$ otherwise. Accordingly, the effective channel from feed port $\mathrm{X}\in\{\mathrm{L},\mathrm{R}\}$ on cable $k$ to user $n$ is given by
\begin{equation}\label{effchannel}
h_{k,n}^\mathrm{X}=\frac{1}{\sqrt{N_k}}\boldsymbol{\alpha}_k^\mathsf{T}\mathbf{h}_{k,n}^\mathrm{X}=\frac{1}{\sqrt{N_k}}\sum_{m=1}^{M}\alpha_{k,m} h_{k,m,n}^\mathrm{X},
\end{equation}
where $N_k = \sum_{m=1}^{M} \alpha_{k,m}$ is the number of activated slots on cable $k$. In the considered system, all radiation slots have identical aperture sizes and coupling characteristics, with each slot restricted to a binary activation state. Therefore, the signal injected from a given feed port is equally divided in power among the activated slots, which leads to the normalization factor $1/\sqrt{N_k}$, reflecting the equal division of the port power among the $N_k$ activated slots. Nevertheless, due to distance-dependent intra-cable attenuation, which is incorporated into $h_{k,m,n}^\mathrm{X}$, the actual radiated powers of the activated slots are generally unequal.

The physical interpretation of the slot activation vector in the proposed LCX based architecture is summarized in the following remark.
\begin{remark}
In the LCX based pinching-antenna system, the slot activation vector serves as an analog beamformer. Activating a slot enables a portion of the guided wave to radiate into free space, with the radiated signal inheriting its amplitude and phase from the intra-cable propagation. Accordingly, the contribution of each slot depends on its position relative to the users and the associated amplitude attenuation and phase variations along the cable and through the wireless channel.
\end{remark}

Since the signals transmitted from ports $\mathrm{L}$ and $\mathrm{R}$ experience different phase variations along the cable, the resulting port-dependent phase responses are exploited through slot activation. Specifically, the slots are selectively activated to enable constructive combining of the signal components from the assigned feed port and destructive combining of those from the other feed ports, thereby enhancing the desired signal while suppressing interference. Therefore, the desired signal of each user is transmitted through only one selected port on a single cable. The transmit signal from port $\mathrm{X}\in\{\mathrm{L}, \mathrm{R}\}$ of cable $k$ is given by
\begin{equation}
x_k^\mathrm{X}=\sum_{n=1}^{N}\beta_{k,n}^\mathrm{X}\sqrt{P_t p_n}s_n,
\end{equation}
where $\beta_{k,n}^\mathrm{X}$ is the port selection indicator, $P_t$ is the total available transmit power, $p_n$ is the power allocation coefficient for user $n$, and $s_n$ is the desired information symbol of user $n$. In particular, $\beta_{k,n}^\mathrm{X}=1$ indicates that the signal intended for user $n$ is transmitted through port $\mathrm{X}$ of cable $k$, whereas $\beta_{k,n}^\mathrm{X}=0$ otherwise. Moreover, the transmitted symbols satisfy $\mathbb{E}\!\left[|s_n|^2\right]=1, \forall n\in\mathcal{N}$.

At user $n$, the signals transmitted from all feed ports can be received as follows:
\begin{align}
y_n^\mathrm{A}&=\sum_{k=1}^{K}\left(h_{k,n}^\mathrm{L} x_k^\mathrm{L}+h_{k,n}^\mathrm{R}x_k^\mathrm{R}\right)+n_0 \\\nonumber
&=\sum_{k=1}^{K}\!\sqrt{\frac{P_t}{N_k}}\!\sum_{i=1}^{N}\left(\beta_{k,i}^\mathrm{L}\boldsymbol{\alpha}_k^\mathsf{T}\mathbf{h}_{k,n}^\mathrm{L}+\beta_{k,i}^\mathrm{R}\boldsymbol{\alpha}_k^\mathsf{T}\mathbf{h}_{k,n}^\mathrm{R}\right)\sqrt{p_i}s_i+n_0,
\end{align}
where $n_0 \sim \mathcal{CN}(0,\sigma^2)$ is the additive white Gaussian noise (AWGN), and $\sigma^2$ is the corresponding noise power.

Since each user is served by only one selected port and coherent combining across different cables or feed ports is not considered, the desired signal and interference powers are evaluated as the sum of the corresponding per-link received powers. Accordingly, the received signal-to-interference-plus-noise ratio (SINR) at user $n$ is given by
\begin{equation}
\gamma_n^\mathrm{A}=\frac{p_n\!\sum_{k=1}^{K}\!\frac{P_t}{N_k}\!\sum_{\mathrm{X}\in\{\mathrm{L},\mathrm{R}\}}\!\beta_{k,n}^\mathrm{X}\!\big|\boldsymbol{\alpha}_k^\mathsf{T}\mathbf{h}_{k,n}^\mathrm{X}\big|^2}{\sum_{i=1,i\neq n}^{N}\!p_i\!\sum_{k=1}^{K}\!\frac{P_t}{N_k}\!\sum_{\mathrm{X}\in\{\mathrm{L},\mathrm{R}\}}\!\beta_{k,i}^\mathrm{X}\!\big|\boldsymbol{\alpha}_k^\mathsf{T}\mathbf{h}_{k,n}^\mathrm{X}\big|^2\!+\!\sigma^2}.
\end{equation}
The corresponding achievable data rate of user $n$ is
\begin{equation}
R_n^\mathrm{A}=\log_2\!\left(1+\gamma_n^\mathrm{A}\right).
\end{equation}
\subsection{Hybrid Beamforming with Digital Precoding}
In the proposed LCX based generalized pinching-antenna system, the two feed ports on each cable can be regarded as independent RF chains. With $K$ cables deployed, the downlink transmission can therefore be modeled as a multiple-input single-output (MISO) system with $2K$ RF chains, in which the information signals of all users are jointly precoded at baseband and simultaneously transmitted through all feed ports on all cables. For user $n$, the corresponding channel vector can be obtained by stacking the effective channels from all feed ports to this user, given by
\begin{equation}\label{channelvec}
\mathbf{h}_n=\left[h_{1,n}^\mathrm{L}\;\;h_{1,n}^\mathrm{R}\;\;h_{2,n}^\mathrm{L}\;\;h_{2,n}^\mathrm{R}\;\;\cdots\;\;h_{K,n}^\mathrm{L}\;\;h_{K,n}^\mathrm{R}\right]\in\mathbb{C}^{1\times 2K}.
\end{equation}
By stacking the channel vectors of all users, the overall channel matrix can be expressed as follows:
\begin{equation}\label{channelmatrix}
\mathbf{H}=\left[\mathbf{h}_1\;\;\mathbf{h}_2\;\;\cdots\;\;\mathbf{h}_N\right]^\mathsf{T}\in\mathbb{C}^{N\times 2K}.
\end{equation}
It is worth noting that the channel matrix inherently incorporates the effects of LCX attenuation, phase progression along the cable, radiation characteristics of the slots, and the slot activation across all cables.

At baseband, ZF digital beamforming can be employed to suppress inter-user interference. The ZF precoding matrix is designed as follows:
\begin{equation}
\widetilde{\mathbf{W}}=\mathbf{H}^\mathsf{H}(\mathbf{H}\mathbf{H}^\mathsf{H})^{-1}=\left[\widetilde{\mathbf{w}}_1\;\;\widetilde{\mathbf{w}}_2\;\;\cdots\;\;\widetilde{\mathbf{w}}_N\right]\in\mathbb{C}^{2K\times N},
\end{equation}
where $\widetilde{\mathbf{w}}_n$ is the unnormalized digital precoding vector for user $n$, given by
\begin{equation}
\widetilde{\mathbf{w}}_n=\left[\tilde{w}_{1,n}^\mathrm{L}\;\;\tilde{w}_{1,n}^\mathrm{R}\;\;\tilde{w}_{2,n}^\mathrm{L}\;\;\tilde{w}_{2,n}^\mathrm{R}\;\;\cdots\;\;\tilde{w}_{K,n}^\mathrm{L}\;\;\tilde{w}_{K,n}^\mathrm{R}\right]^\mathsf{T}\in\mathbb{C}^{2K\times 1}.
\end{equation}
To facilitate the enforcement of the transmit power constraint, each precoding vector is normalized as follows:
\begin{equation}
\mathbf{w}_n=\frac{\widetilde{\mathbf{w}}_n}{\|\widetilde{\mathbf{w}}_n\|},\;\forall n\in\mathcal{N},
\end{equation}
and the resulting normalized ZF precoding matrix is given by
\begin{equation}
\mathbf{W}=\left[\mathbf{w}_1\;\;\mathbf{w}_2\;\;\cdots\;\;\mathbf{w}_N\right]\in\mathbb{C}^{2K\times N}.
\end{equation}
Accordingly, the ZF precoder satisfies $\mathbf{h}_n\mathbf{w}_n=1/\|\widetilde{\mathbf{w}}_n\|,\forall n$ and $\mathbf{h}_n\mathbf{w}_i=0, \forall i\neq n$. 

Based on the ZF precoding matrix, the transmit signal vector across all cable feed ports can be expressed as follows:
\begin{equation}
\mathbf{x}=\sum_{n=1}^{N}\mathbf{w}_n\sqrt{P_t p_n}s_n,
\end{equation}
which satisfies $\mathbb{E}[\|\mathbf{x}\|^2]=P_t\sum_{n=1}^N p_n$. With the ZF based hybrid beamforming architecture, inter-user interference is completely eliminated. Accordingly, the received signal at user $n$ can be written as follows:
\begin{equation}
y_n^\mathrm{H}=\mathbf{h}_n\mathbf{x}+n_0=\sqrt{P_t}\mathbf{h}_n\mathbf{w}_n\sqrt{p_n}s_n+n_0 .
\end{equation}
The SINR experienced by user $n$ is given by
\begin{equation}\label{hbsinr}
\gamma_n^\mathrm{H}=\frac{p_n P_t\left|\mathbf{h}_n \mathbf{w}_n\right|^2}{\sigma^2},
\end{equation}
and the corresponding achievable data rate is expressed as
\begin{equation}
R_n^\mathrm{H}=\log_2\!\left(1+\gamma_n^\mathrm{H}\right).
\end{equation}
\section{Problem Formulation}
The objective of the proposed LCX based pinching-antenna system with dual-port feeding is to ensure user fairness by efficiently exploiting the available spatial degrees of freedom. In particular, since the analog and hybrid transmission schemes exhibit different interference characteristics, minimum rate maximization is selected to provide a common fairness oriented criterion for evaluating both schemes. To this end, the minimum achievable data rate is adopted as the system performance metric \cite{tegos2025pin}. Based on the developed signal models for analog and hybrid beamforming, two minimum rate maximization problems are formulated. Specifically, the first problem considers slot activation, port selection, and power allocation for analog beamforming, whereas the second problem optimizes slot activation and power allocation for hybrid beamforming with digital precoding.

The minimum rate maximization problem for the analog beamforming scheme is formulated as follows:
\begin{subequations}
\begin{empheq}{align}
\max_{\boldsymbol{\alpha},\boldsymbol{\beta},\mathbf{p}}\quad & \min\{R_n^\mathrm{A}| n\in\mathcal{N}\}\\
\textrm{s.t.} \quad & \alpha_{k,m}\in\{0,1\}, \forall m\in\mathcal{M}_k, \forall k\in\mathcal{K},\\
& \sum\nolimits_{m=1}^M\alpha_{k,m}\ge 1, \forall k\in\mathcal{K},\\
& \beta_{k,n}^\mathrm{X}\!\in\{0,1\}, \forall \mathrm{X}\in\{\mathrm{L},\mathrm{R}\}, \forall n\in\mathcal{N}, \forall k\in\mathcal{K},\\
& \sum\nolimits_{k=1}^K(\beta_{k,n}^\mathrm{L}+\beta_{k,n}^\mathrm{R}) = 1, \forall n\in\mathcal{N},\\
& \sum\nolimits_{n=1}^N\beta_{k,n}^\mathrm{X}= 1, \forall \mathrm{X}\in\{\mathrm{L},\mathrm{R}\}, \forall k\in\mathcal{K},\\
& p_n\ge 0, \forall n\in\mathcal{N},\\
& \sum\nolimits_{n=1}^N p_n\le 1,
\end{empheq}
\label{abproblem}
\end{subequations}\vspace{-2mm}\\
where $\boldsymbol{\alpha}$, $\boldsymbol{\beta}$, and $\mathbf{p}$ denote the sets of slot activation indicators, port selection indicators, and power allocation coefficients, respectively. Constraint~(\ref{abproblem}c) ensures that at least one slot is activated on each cable. Constraints~(\ref{abproblem}e) and~(\ref{abproblem}f) enforce a one-to-one association between users and feed ports, i.e., each user is assigned to exactly one port, and each port serves exactly one user.

For the hybrid beamforming architecture with digital precoding, the corresponding minimum rate maximization problem is formulated as follows:
\begin{subequations}
\begin{empheq}{align}
\max_{\boldsymbol{\alpha},\mathbf{p}}\quad & \min\{R_n^\mathrm{H}| n\in\mathcal{N}\}\\
\textrm{s.t.} \quad & \alpha_{k,m} \in\{0,1\}, \forall m \in\mathcal{M}_k, \forall k\in\mathcal{K},\\
& \sum\nolimits_{m=1}^M\alpha_{k,m}\ge 1, \forall k\in\mathcal{K},\\
& p_n\ge 0, \forall n\in\mathcal{N},\\
& \sum\nolimits_{n=1}^Np_n\le 1.
\end{empheq}
\label{hbproblem}
\end{subequations}\vspace{-2mm}\\
In the above problem, constraint~(\ref{hbproblem}c) ensures that at least one slot is activated on each cable, while constraints~(\ref{hbproblem}d) and~(\ref{hbproblem}e) impose the non-negativity and sum-power constraints on the power allocation coefficients, respectively.
\section{Solution for Analog Beamforming}
To facilitate the solution of the mixed-integer optimization problem in \eqr{abproblem}, the original problem is decoupled into two subproblems. Specifically, the first subproblem optimizes the discrete slot activation and port selection variables for a given power allocation, while the second subproblem determines the optimal power allocation under fixed slot activation and port selection. Accordingly, the slot activation and port selection subproblem is presented as follows:
\begin{subequations}
\begin{empheq}{align}
\max_{\boldsymbol{\alpha},\boldsymbol{\beta}}\quad & \min\{R_n^\mathrm{A}| n\in\mathcal{N}\}\\\nonumber
\textrm{s.t.} \quad & \text{(\ref{abproblem}b)}, \text{(\ref{abproblem}c)}, \text{(\ref{abproblem}d)}, \text{(\ref{abproblem}e)}, \text{and (\ref{abproblem}f)},
\end{empheq}
\label{abproblem1}
\end{subequations}\vspace{-2mm}\\
and the power allocation subproblem is given by
\begin{subequations}
\begin{empheq}{align}
\max_{\mathbf{p}}\quad & \min\{R_n^\mathrm{A}| n\in\mathcal{N}\}\\\nonumber
\textrm{s.t.} \quad & \text{(\ref{abproblem}g)}, \text{and (\ref{abproblem}h)}.
\end{empheq}
\label{abproblem2}
\end{subequations}\vspace{-2mm}\\
The corresponding solution methods for the resulting subproblems are introduced in the following subsections.
\subsection{Joint Port Selection and Slot Activation via Matching and Coalitional Games}
Since each user is assigned to exactly one port and each port serves exactly one user, the port selection problem in \eqr{abproblem1} can be modeled as a perfect matching between the user set and the port set \cite{burkard2012assignment}. Specifically, the port set is defined as $\mathcal{P}=\{(k,\mathrm{L}), (k,\mathrm{R})| k\!\in\!\mathcal{K}\}$, which satisfies $|\mathcal{P}|=|\mathcal{N}|=2K$. Based on constraints~(\ref{abproblem}e) and~(\ref{abproblem}f), the user-port association establishes a bijective mapping between $\mathcal{N}$ and $\mathcal{P}$. This mapping can be formally characterized as a perfect matching \cite{lovasz2009matching}, defined as follows.
\begin{definition}
A perfect matching $\mu$ is a mapping from the user set $\mathcal{N}$ to the port set $\mathcal{P}$, i.e., $\mu:\mathcal{N}\rightarrow\mathcal{P}$, which satisfies:
\begin{enumerate}[leftmargin=*]
\item $\mu(n)\in\mathcal{P}$, $\forall n\in\mathcal{N}$, and $\mu^{-1}(k,\mathrm{X})\in\mathcal{N}$, $\forall (k,\mathrm{X})\in\mathcal{P}$;
\item $|\mu(n)|=1$, $\forall n\in\mathcal{N}$, and $|\mu^{-1}(k,\mathrm{X})|=1$, $\forall (k,\mathrm{X})\in\mathcal{P}$;
\item $\mu(n)=(k,\mathrm{X})\Leftrightarrow \mu^{-1}(k,\mathrm{X})=n$.
\end{enumerate}
\end{definition}

Based on the above definition, the perfect matching $\mu$ can uniquely determine port selection indicators. Specifically, the port selection indicator can be expressed as follows:
\begin{equation}
\beta_{k,n}^\mathrm{X}=
\begin{cases}
1, & \text{if } \mu(n)=(k,\mathrm{X}),\\
0, & \text{otherwise}.
\end{cases}
\end{equation}
As a result, the binary constraint in (\ref{abproblem}d) and the assignment constraints in (\ref{abproblem}e) and (\ref{abproblem}f) are guaranteed by the perfect matching structure.

When a user intends to change its assigned port, a direct reassignment is not feasible due to the one-to-one association constraint. Instead, the user must exchange its assigned port with another user occupying the target port. This leads to a swap operation involving two users, defined as follows.
\begin{definition}
Given a matching $\mu$, a swap operation between users $n$ and $n'$ assigned to ports $\mu(n)=(k,\mathrm{X})$ and $\mu(n')=(k',\mathrm{X}')$ results in a new matching $\mu_{n'}^n$ defined by
\begin{equation}
\mu_{n'}^n(n) = (k',\mathrm{X}'),
\end{equation}
and
\begin{equation}
\mu_{n'}^n(n') = (k,\mathrm{X}),
\end{equation}
while the assignments of all other users remain unchanged. The resulting matching $\mu_{n'}^{n}$ is also a perfect matching and therefore preserves the feasibility of the user-port association.
\end{definition}

To facilitate the swap operation, a preference relation over perfect matchings is constructed. In the considered system, the achievable data rate of each user depends on the entire matching through inter-user interference, which introduces externalities among users. As a result, individual user preferences are insufficient to characterize the impact of a swap operation. To align with the system performance objective, the minimum achievable data rate is adopted as the global utility metric. The preference relation over matchings is therefore defined as
\begin{equation}
\mu \prec \mu_{n'}^n \Leftrightarrow  \min\{R_i^\mathrm{A}(\mu)|i\in\mathcal{N}\} < \min\{R_i^\mathrm{A}(\mu_{n'}^n)|i\in\mathcal{N}\}.
\end{equation}
In other words, a swap operation between users $n$ and $n'$ is accepted if and only if the minimum achievable data rate is strictly increased. Otherwise, the swap is rejected and the original matching is retained.

When a user is reassigned to a different port, the cables for transmitting the desired signal and the interference signals may change. Hence, slot activation on the affected cables needs to be updated in order to evaluate the resulting minimum data rate. In this context, the slot activation problem on cable $k$ can be modeled as a coalitional game $(\mathcal{M}_k, v, \mathcal{S}_k)$, where $v$ is the coalition value. In this game, each slot can be activated or deactivated by joining or leaving coalition $\mathcal{S}_k$, respectively. Defining the coalition structure as $\mathcal{S}=\{\mathcal{S}_1, \mathcal{S}_2, \dots, \mathcal{S}_K\}$, the activation of slot $m$ on cable $k$, where $m\notin\mathcal{S}_k$, leads to an updated coalition structure given by
\begin{equation}
\mathcal{S}'=\mathcal{S}\backslash\mathcal{S}_k\cup\{\mathcal{S}_k\cup\{m\}\}.
\end{equation}
Similarly, if slot $m$ on cable $k$ is deactivated, the coalition structure is updated as follows:
\begin{equation}
\mathcal{S}'=\mathcal{S}\backslash\mathcal{S}_k\cup\{\mathcal{S}_k\backslash\{m\}\}.
\end{equation}

Since the activation or deactivation of any slot affects the received signals of all users, the coalition value is defined as the minimum achievable data rate, i.e.,
\begin{equation}
v(\mathcal{S})=\min\{R_i^\mathrm{A}(\mathcal{S})|i\in\mathcal{N}\}.
\end{equation}
Accordingly, a slot update is accepted only if the resulting coalition structure leads to a non-decreasing coalition value, as follows:
\begin{equation}
\mathcal{S}\prec \mathcal{S}'\Leftrightarrow v(\mathcal{S}) <v(\mathcal{S}').
\end{equation}

Based on the proposed matching framework with an embedded coalitional game, a joint port selection and slot activation algorithm is presented in Algorithm~\ref{cgalg}.

\begin{algorithm}[t]
\caption{Port Selection and Slot Activation Algorithm}
\label{cgalg}
\begin{algorithmic}[1]
\STATE \textbf{Initialization:}
\STATE Randomly assign all users to ports to obtain $\mu$.
\STATE Activate the nearest slot for each user to obtain $\mathcal{S}$.
\STATE \textbf{Main Loop:}
\FOR{$n\in\mathcal{N}$, where $R_n^\mathrm{A}=\min\{R_i^\mathrm{A}|i\in\mathcal{N}\}$}
\STATE Find $(k,\mathrm{X})=\mu(n)$.
\FOR{$(k',\mathrm{X}')\in\mathcal{P}$}
\IF{$\mu(n)\neq (k',\mathrm{X}')$}
\STATE Find $n'=\mu^{-1}(k',\mathrm{X}')$.
\STATE Generate $\mu_{n'}^n$.
\FOR{$i\in\{k,k'\}$}
\FOR{$m\in\mathcal{M}_i$}
\IF{$m\in\mathcal{S}_i$}
\STATE $\mathcal{S}'=\mathcal{S}\backslash\mathcal{S}_i\cup\{\mathcal{S}_i\backslash\{m\}\}$.
\ELSE
\STATE $\mathcal{S}'=\mathcal{S}\backslash\mathcal{S}_i\cup\{\mathcal{S}_i\cup\{m\}\}$
\ENDIF
\IF{$\mathcal{S}\prec\mathcal{S}'$}
\STATE $\mathcal{S}=\mathcal{S}'$.
\ENDIF
\ENDFOR
\ENDFOR
\IF{$\mu\prec\mu_{n'}^n$}
\STATE $\mu=\mu_{n'}^n$.
\ENDIF 
\ENDIF 
\ENDFOR
\ENDFOR
\end{algorithmic}
\end{algorithm}

In Algorithm~\ref{cgalg}, the developed power allocation solution can be incorporated into the evaluation of the achievable data rates, which serve as the matching utilities and coalition values. In each iteration, the user with the minimum achievable data rate is selected under equal power allocation, as shown in line~5, thereby identifying the user experiencing the worst channel conditions. For the selected user, all feasible swap operations with other users are examined, and the corresponding slot activation on the affected cables is updated through the coalitional game mechanism, as illustrated in lines~11-22. A swap operation is accepted only if it results in a strict increase in the minimum data rate. The algorithm terminates when no user can further improve the system performance through any swap operation over a complete iteration cycle. The proposed port selection and slot activation framework can also be extended to sum rate maximization by redefining the matching utilities and coalition values according to the sum-rate objective, while retaining the same swap-based local improvement procedure.

The proposed algorithm adopts a local improvement strategy in which user assignment and slot activation updates are accepted only if they strictly increase the minimum achievable data rate. Since the numbers of feasible user-port matchings and slot activation configurations are finite, and the minimum achievable data rate is upper bounded due to limited transmit power and channel gains, the algorithm cannot admit an infinite sequence of improving updates and is therefore guaranteed to converge in a finite number of iterations.

In the worst case, the user with the minimum achievable data rate needs to examine all possible port exchanges, resulting in $2K-1$ candidate swap operations. For each candidate swap, the slot activation on at most two cables is locally updated. On each affected cable, all $M$ slots are sequentially examined for possible activation or deactivation, leading to $2M$ slot update trials. Therefore, the total number of slot-update trials per iteration is upper bounded by $2(2K-1)M$. Let $T$ denote the total number of outer iterations. By ignoring constant factors, the overall computational complexity of Algorithm~\ref{cgalg} is given by $\mathcal{O}(TKM)$.
\subsection{Bisection based Power Allocation}
In this subsection, the power allocation problem is solved for given port selection and slot activation configurations. The following lemma establishes an intrinsic property of the optimal power allocation that holds under both the analog and hybrid beamforming architectures.
\begin{lemma}\label{lemma1}
Consider a minimum data rate maximization problem as follows:
\begin{subequations}
\begin{empheq}{align}
\max_{\mathbf{p}}\quad & \min\{R_n(\mathbf{p})|n\in\mathcal N\}\\
\textrm{s.t.} \quad & p_n\ge 0, \forall n\in\mathcal{N},\\
& \sum\nolimits_{n=1}^N p_n\le 1.
\end{empheq}
\label{paproblem}
\end{subequations}\vspace{-2mm}\\
There exists an optimal solution $\mathbf{p}^\star$ such that all users achieve the same data rate, i.e.,
\begin{equation}\label{samerate}
R_1(\mathbf{p}^\star)=R_2(\mathbf{p}^\star)=\cdots=R_N(\mathbf{p}^\star).
\end{equation}
\end{lemma}
\begin{IEEEproof}
Let $\mathbf{p}$ be an optimal feasible solution of \eqr{paproblem}, and define $R_{\min}(\mathbf{p})\!=\!\min\{R_n(\mathbf{p})|n\!\in\!\mathcal{N}\}$. If \eqr{samerate} does not hold, there exists a bottleneck user $i\!\in\!\arg\min\{R_n(\mathbf{p})|n\!\in\!\mathcal{N}\}$, and another user $j$ such that $R_j(\mathbf{p})\!>\!R_i(\mathbf{p})$. Since the data rate of each user is strictly increasing with respect to its own power allocation coefficient when the power allocations of the other users are fixed, a small amount of power $\Delta p\!>\!0$ can be transferred from user $j$ to user $i$ while preserving feasibility, resulting in $p_i'=p_i+\Delta p$ and $p_j'=p_j-\Delta p$. 

In this case, $R_i(\mathbf{p}')$ strictly increases, while $R_j(\mathbf{p}')$ decreases. For sufficiently small $\Delta p$, it holds that $R_j(\mathbf{p}')>R_i(\mathbf{p})$. Since $R_n(\mathbf{p})$ is continuous in $\mathbf{p}$ for the analog beamforming scheme, one can choose $\Delta p$ sufficiently small such that $R_n(\mathbf{p}')\ge R_i(\mathbf{p}), \forall n\notin\{i,j\}$. For the hybrid beamforming scheme with zero-forcing precoding, the rates of all remaining users remain unchanged.

As a result, the minimum data rate satisfies $R_{\min}(\mathbf{p}')>R_{\min}(\mathbf{p})$, which contradicts the optimality of $\mathbf{p}$. Therefore, there exists an optimal solution in which all users achieve the same data rate, which completes the proof.
\end{IEEEproof}

Based on Lemma~\ref{lemma1}, the optimal power allocation can be obtained by searching for the maximum common data rate, leading to the following bisection based solution.

Given $\boldsymbol{\alpha}$ and $\boldsymbol{\beta}$, the link gains for the desired signal and interference signal can be respectively represented as follows:
\begin{equation}
a_n=\sum_{k=1}^K\frac{P_t}{N_k}\!\sum_{\mathrm{X}\in\{\mathrm{L},\mathrm{R}\}}\!\beta_{k,n}^\mathrm{X}|\boldsymbol{\alpha}_k^\mathsf{T}\mathbf{h}_{k,n}^\mathrm{X}|^2,
\end{equation}
and
\begin{equation}
b_{n,i}=\sum_{k=1}^K\frac{P_t}{N_k}\!\sum_{\mathrm{X}\in\{\mathrm{L},\mathrm{R}\}}\!\beta_{k,i}^\mathrm{X}|\boldsymbol{\alpha}_k^\mathsf{T}\mathbf{h}_{k,n}^\mathrm{X}|^2, \forall i\neq n.
\end{equation}
Accordingly, the SINR of user $n$ can be rewritten as follows:
\begin{equation}
\gamma_n^\mathrm{A}(\mathbf{p})=\frac{p_na_n}{\sum_{i=1,i\neq n}^Np_ib_{n,i}+\sigma^2}.
\end{equation}
Since $R_n^\mathrm{A}=\log_2(1+\gamma_n^\mathrm{A})$ is monotonically increasing in $\gamma_n^\mathrm{A}$, the power allocation subproblem in \eqr{abproblem2} is equivalent to the following max-min SINR optimization:
\begin{subequations}
\begin{empheq}{align}
\max_{\mathbf{p}}\quad & \min\{\gamma_n^\mathrm{A}(\mathbf{p})| n\in\mathcal{N}\}\\\nonumber
\textrm{s.t.} \quad & \text{(\ref{abproblem}g)}, \text{and (\ref{abproblem}h)}.
\end{empheq}
\label{abproblem21}
\end{subequations}\vspace{-2mm}

By introducing a SINR target $\tau$, where $\tau \ge 0$, problem~\eqr{abproblem21} can be equivalently reformulated as follows:
\begin{subequations}
\begin{empheq}{align}
\max_{\mathbf{p},\tau}\quad & \tau\\
\textrm{s.t.} \quad & p_na_n\ge \tau\!\!\!\sum_{i=1,i\neq n}^Np_ib_{n,i}+\tau\sigma^2, \forall n \in\mathcal{N},\\\nonumber
& \text{(\ref{abproblem}g)}, \text{and (\ref{abproblem}h)}.
\end{empheq}
\label{abproblem22}
\end{subequations}\vspace{-2mm}\\
In problem \eqr{abproblem22}, constraint~(\ref{abproblem22}b) is obtained by rewriting $\gamma_n^\mathrm{A}(\mathbf{p}) \ge \tau$ into an equivalent inequality form. For any $\tau \ge \tau'$, the condition $\gamma_n^\mathrm{A}(\mathbf{p}) \ge \tau$ implies $\gamma_n^\mathrm{A}(\mathbf{p}) \ge \tau'$. As a result, the feasible set of problem~\eqr{abproblem22} exhibits a monotonic structure.

By defining a non-negative matrix $\mathbf{F}\in \mathbb{R}_+^{N\times N}$ and a vector
$\mathbf{u}\in\mathbb{R}_+^{N\times 1}$ as 
\begin{equation}
[\mathbf{F}]_{n,i}=
\begin{cases}
\frac{b_{n,i}}{a_n}, & i\neq n,\\
0, & i=n,
\end{cases}
\end{equation}
and 
\begin{equation}
[\mathbf{u}]_n=\frac{\sigma^2}{a_n},
\end{equation}
problem~\eqr{abproblem22} can be further reformulated as follows:
\begin{subequations}
\begin{empheq}{align}
\max_{\mathbf{p},\tau}\quad & \tau\\
\textrm{s.t.} \quad & \mathbf{p}\succeq\tau\mathbf{F}\mathbf{p}+\tau\mathbf{u},\\
& \mathbf{p}\succeq 0,\\
& \mathbf{1}^{\mathsf T}\mathbf{p}\le 1.
\end{empheq}
\label{abproblem23}
\end{subequations}\vspace{-2mm}\\
Constraint~(\ref{abproblem23}b) implies that, in order to satisfy the target SINR $\tau$, the power allocated to each user must be at least a term proportional to the aggregate interference, i.e., $\tau\mathbf{F}\mathbf{p}$, plus an additional noise related term given by $\tau\mathbf{u}$.

If $\mathbf{I}-\tau\mathbf{F}$ is invertible and its inverse is non-negative, the minimum power allocation coefficients required to satisfy constraint~(\ref{abproblem23}b) can be expressed as follows:
\begin{equation}
\mathbf{p}_\mathrm{min}(\tau)=\tau(\mathbf{I}-\tau\mathbf{F})^{-1}\mathbf{u}.
\end{equation}
It is noted that the above expression is valid when $\rho(\tau\mathbf{F})<1$, where $\rho(\cdot)$ denotes the spectral radius. By incorporating constraint~(\ref{abproblem23}d), the SINR target $\tau$ is feasible if and only if
\begin{equation}
\mathbf{1}^{\mathsf T}\mathbf{p}_{\min}(\tau)\le 1.
\end{equation}
Based on the monotonic feasibility, a bisection based algorithm is proposed in Algorithm~\ref{bisection_power} to obtain the optimal solution.

\begin{algorithm}[t]
\caption{Bisection based Algorithm for Solving \eqr{abproblem23}}
\label{bisection_power}
\begin{algorithmic}[1]
\STATE Construct $\mathbf{F}$ and $\mathbf{u}$ with $a_n$ and $b_{n,i}$.
\STATE Set $\tau_\mathrm{min}$, $\tau_\mathrm{max}$, and $\epsilon>0$.
\WHILE{$\tau_\mathrm{max}-\tau_\mathrm{min}>\epsilon$}
\STATE $\tau=(\tau_\mathrm{min}+\tau_\mathrm{max})/2$.
\STATE Solve $(\mathbf{I}-\tau\mathbf{F})\mathbf{x}=\mathbf{u}$ for $\mathbf{x}$.
\STATE $\mathbf{p}=\tau \mathbf{x}$.
\IF{$\mathbf{p}\succeq \mathbf{0}$ \textbf{and} $\mathbf{1}^{\mathsf T}\mathbf{p}\le 1$}
\STATE $\tau_\mathrm{min}=\tau$, $\mathbf{p}^\star=\mathbf{p}$ .
\ELSE
\STATE $\tau_\mathrm{max}=\tau$.
\ENDIF
\ENDWHILE
\end{algorithmic}
\end{algorithm}

Since the feasibility is monotonic in $\tau$, the set of feasible SINR targets forms an interval $[0,\tau^\star]$, where $\tau^\star$ denotes the optimal value. Specifically, for any $\tau<\tau^\star$, the SINR target is feasible and yields a valid power allocation $\mathbf{p}_{\min}(\tau)$, whereas for any $\tau>\tau^\star$, the corresponding problem becomes infeasible. As a result, $\mathbf{p}_{\min}(\tau^\star)$ is feasible and satisfies $\min_{n}\gamma_n(\mathbf{p})\ge\tau^\star$. That is, by the definition of $\tau^\star$, no power allocation can achieve a larger minimum SINR. Therefore, $\mathbf{p}^\star=\mathbf{p}_{\min}(\tau^\star)$ is globally optimal for problem~\eqr{abproblem2}.
\section{Solution for Hybrid Beamforming}
In this section, the hybrid beamforming design is investigated by decomposing problem~\eqr{hbproblem} into two subproblems, where each subproblem is solved with the other being fixed. The slot activation subproblem is formulated as
\begin{subequations}
\begin{empheq}{align}
\max_{\boldsymbol{\alpha}}\quad & \min\{R_n^\mathrm{H}| n\in\mathcal{N}\}\\\nonumber
\textrm{s.t.} \quad & \text{(\ref{hbproblem}b)}, \text{and (\ref{hbproblem}c)},
\end{empheq}
\label{hbproblem1}
\end{subequations}\vspace{-2mm}\\
and the corresponding power allocation subproblem is
\begin{subequations}
\begin{empheq}{align}
\max_{\mathbf{p}}\quad & \min\{R_n^\mathrm{H}| n\in\mathcal{N}\}\\\nonumber
\textrm{s.t.} \quad & \text{(\ref{hbproblem}d)}, \text{and (\ref{hbproblem}e)}.
\end{empheq}
\label{hbproblem2}
\end{subequations}\vspace{-2mm}\\
The solution approaches for the above two subproblems are presented in the following subsections.
\subsection{Slot Activation via Coalitional Games}
In the hybrid beamforming scenario, the inter-user interference is completely eliminated by the ZF digital precoding. As a result, the slot activation problem in \eqr{hbproblem1} can be reformulated as an equivalent max-min effective channel gain optimization, which facilitates the subsequent coalitional game based design.

For a given slot activation $\boldsymbol{\alpha}$, the effective channel matrix $\mathbf{H}$ is determined according to \eqr{effchannel}, \eqr{channelvec}, and \eqr{channelmatrix}. Under the ZF precoding, the conditions $\mathbf{h}_n\mathbf{w}_n = 1/\|\widetilde{\mathbf{w}}_n\|,\forall n$, and $\mathbf{h}_n\mathbf{w}_i = 0,\forall i\neq n$, are satisfied. Accordingly, the received SINR of user $n$ in \eqr{hbsinr} can be rewritten as
\begin{equation}
\gamma_n^\mathrm{H}=\frac{p_n P_t}{\sigma^2}|\mathbf{h}_n\mathbf{w}_n|^2=\frac{p_n P_t}{\sigma^2}\frac{1}{\|\widetilde{\mathbf{w}}_n\|^2}.
\end{equation}
Moreover, it holds that $\|\widetilde{\mathbf{w}}_n\|^2=\big[(\mathbf{H}\mathbf{H}^\mathsf{H})^{-1}\big]_{n,n}$. Therefore, the effective channel gain of user $n$ can be defined as
\begin{equation}
\zeta_n(\boldsymbol{\alpha})\triangleq |\mathbf{h}_n\mathbf{w}_n|^2=\frac{1}{\|\widetilde{\mathbf{w}}_n\|^2}=\frac{1}{\left[(\mathbf{H}\mathbf{H}^\mathsf{H})^{-1}\right]_{n,n}}.
\end{equation}
Accordingly, the achievable data rate of user $n$ is given by
\begin{equation}
R_n^\mathrm{H}=\log_2\!\left(1+\frac{p_n P_t}{\sigma^2}\zeta_n(\boldsymbol{\alpha})\right).
\end{equation}

Since $\log_2(1+x)$ is strictly increasing for $x\ge 0$, maximizing the minimum achievable rate is equivalent to maximizing the minimum SINR, which further reduces to maximizing the minimum weighted effective channel gain. Accordingly, problem~\eqr{hbproblem1} can be equivalently reformulated as follows:
\begin{subequations}
\begin{empheq}{align}
\max_{\boldsymbol{\alpha}}\quad & \min\{p_n\zeta_n(\boldsymbol{\alpha})|n\in\mathcal{N}\}\\\nonumber
\textrm{s.t.} \quad & \text{(\ref{hbproblem}b)}, \text{and (\ref{hbproblem}c)}.
\end{empheq}
\label{hbproblem11}
\end{subequations}\vspace{-2mm}\\
When solving problem~\eqr{hbproblem1}, the power allocation vector $\mathbf{p}$ is treated as fixed. Hence, problem~\eqr{hbproblem11} reduces to an unweighted max-min effective channel gain optimization, given by
\begin{subequations}
\begin{empheq}{align}
\max_{\boldsymbol{\alpha}}\quad & \min\{\zeta_n(\boldsymbol{\alpha})|n\in\mathcal{N}\}\\\nonumber
\textrm{s.t.} \quad & \text{(\ref{hbproblem}b)}, \text{and (\ref{hbproblem}c)}.
\end{empheq}
\label{hbproblem12}
\end{subequations}\vspace{-2mm}

As indicated by problem~\eqr{hbproblem12}, under the hybrid beamforming architecture, all activated slots jointly determine the effective channel matrix and the resulting ZF precoder. Therefore, the slot activation problem can be modeled as a coalitional formation game, in which all activated slots cooperate by forming a single coalition. The set of players consists of all radiation slots deployed on all cables, given by
\begin{equation}
\mathcal{V} = \{(k,m)| k\in\mathcal{K}, m\in\mathcal{M}_k\}.
\end{equation}
The coalition is defined as a subset $\mathcal{S}\subseteq\mathcal{V}$, which represents the set of activated slots. This coalition uniquely determines the slot activation vector $\boldsymbol{\alpha}$ and satisfies constraint~(\ref{hbproblem}b), where $\alpha_{k,m}=1$ if $(k,m)\in\mathcal{S}$ and $\alpha_{k,m}=0$ otherwise. Moreover, according to constraint~(\ref{hbproblem}c), at least one slot must be activated on each cable, which leads to the following condition:
\begin{equation}
|\mathcal{S}\cap \mathcal{M}_k|\ge 1,\quad \forall k\in\mathcal{K}.
\end{equation}
Based on problem~\eqr{hbproblem12}, the coalition value is defined as the minimum effective channel gain among all users, i.e.,
\begin{equation}
v(\mathcal{S}) \triangleq \min\{\zeta_n(\mathcal{S})\,|\,n\in\mathcal{N}\},
\end{equation}
where $\zeta_n(\mathcal{S})$ is the effective channel gain of user $n$ under coalition $\mathcal{S}$.

In the coalitional game, each slot adopts a merge-and-split strategy \cite{han2012game, kaidi2024nfc}, which is described as follows.
\begin{definition}
Given a player $(k,m)\in\mathcal{V}$ and a coalition $\mathcal{S}$, the following merge-and-split operations are defined.
\begin{enumerate}[leftmargin=*]
\item \textbf{Merge Rule}: Any player $(k,m)\notin\mathcal{S}$ is allowed to merge with the coalition $\mathcal{S}$ to form a new coalition
\begin{equation}
\mathcal{S}'=\mathcal{S}\cup\{(k,m)\},
\end{equation}
if and only if $v(\mathcal{S}')>v(\mathcal{S})$.
\item \textbf{Split Rule}: Any player $(k,m)\in\mathcal{S}$ is allowed to split from the coalition $\mathcal{S}$, resulting in a new coalition
\begin{equation}
\mathcal{S}'=\mathcal{S}\backslash\{(k,m)\},
\end{equation}
if and only if $v(\mathcal{S}')>v(\mathcal{S})$ and $|\mathcal{S}\cap\mathcal{M}_k|>1$.
\end{enumerate}
\end{definition}

Based on the proposed merge-and-split rules, the slot activation procedure is summarized in Algorithm~\ref{cgalg2}.
\begin{algorithm}[t]
\caption{Slot Activation Algorithm}
\label{cgalg2}
\begin{algorithmic}[1]
\STATE \textbf{Initialization:}
\STATE Activate the nearest slot on each cable to obtain $\mathcal{S}$.
\STATE Set $\delta\in\{0,1\}$ as an indicator variable.
\STATE \textbf{Main Loop:}
\FOR{$(k,m)\in\mathcal{V}$}
\STATE $\delta\leftarrow 0$.
\IF{$(k,m)\notin\mathcal{S}$}
\STATE $\mathcal{S}'=\mathcal{S}\cup\{(k,m)\}$.
\STATE $\delta\leftarrow 1$.
\ELSIF{$|\mathcal{S}\cap\mathcal{M}_k|>1$}
\STATE $\mathcal{S}'=\mathcal{S}\backslash\{(k,m)\}$.
\STATE $\delta\leftarrow 1$.
\ENDIF
\IF{$\delta=1$ \textbf{and} $v(\mathcal{S})<v(\mathcal{S}')$}
\STATE $\mathcal{S}=\mathcal{S}'$.
\ENDIF
\ENDFOR
\end{algorithmic}
\end{algorithm}

The slot activation process follows an iterative coalition formation procedure based on local improvement. Starting from an initial feasible coalition, slots sequentially attempt merge or split operations. Each candidate update is accepted only if it leads to a strict increase in the coalition value. The procedure terminates when no further improving merge or split operation exists. The same coalition formation procedure remains applicable to sum rate maximization by redefining the coalition value according to the achievable sum rate.

The number of feasible coalitions is finite, since each cable contains a finite number of non-empty slot subsets. Moreover, the coalition value strictly increases at every accepted update and is upper bounded due to the finite transmit power and channel gains. Therefore, the proposed coalition formation process is guaranteed to converge in a finite number of iterations to a locally stable coalition structure.

Algorithm~\ref{cgalg2} performs slot activation updates by sequentially scanning all available slots. In the worst case, each iteration involves at most one merge or split attempt for every slot, resulting in $KM$ candidate updates per iteration. Let $T$ denote the number of iterations until convergence. The overall computational complexity of the proposed slot activation algorithm therefore scales as $\mathcal{O}(TKM)$. Under hybrid beamforming, the coalition value depends only on the effective channel gains, which eliminates the need for interference and rate evaluations and leads to a significantly reduced computational complexity compared with the analog scheme.
\subsection{Closed-Form Power Allocation}
For the considered hybrid beamforming scheme, the inter-user interference is completely eliminated by the ZF digital precoding. Accordingly, the achievable rate of user $n$ can be expressed as follows:
\begin{equation}
R_n^\mathrm{H}=\log_2(1+c_n p_n),
\end{equation}
where
\begin{equation}
c_n=\frac{P_t|\mathbf{h}_n\mathbf{w}_n|^2}{\sigma^2}.
\end{equation}
Therefore, for a given slot activation $\boldsymbol{\alpha}$, the power allocation subproblem in \eqr{hbproblem2} reduces to
\begin{subequations}
\begin{empheq}{align}
\max_{\mathbf{p}}\quad & \min\{\log_2(1+c_np_n)| n\in\mathcal{N}\}\\\nonumber
\textrm{s.t.} \quad & \text{(\ref{hbproblem}d)}, \text{and (\ref{hbproblem}e)},
\end{empheq}
\label{hbproblem21}
\end{subequations}\vspace{-2mm}

To solve the above problem, the following proposition is obtained.
\begin{proposition}
For any given slot activation result, the optimal solution to
problem~\eqr{hbproblem21} is given by
\begin{equation}\label{optip2}
p_n^\star=\frac{1/c_n}{\sum_{i=1}^N 1/c_i},\quad \forall n\in\mathcal{N}.
\end{equation}
\end{proposition}

\begin{IEEEproof}
Since $\log_2(\cdot)$ is strictly increasing, maximizing $\log_2(1+c_n p_n)$ is equivalent to maximizing the minimum SINR, which leads to the following equivalent problem:
\begin{subequations}
\begin{empheq}{align}
\max_{\mathbf{p}}\quad & \min\{c_np_n| n\in\mathcal{N}\}\\\nonumber
\textrm{s.t.} \quad & \text{(\ref{hbproblem}d)}, \text{and (\ref{hbproblem}e)}.
\end{empheq}
\label{hbproblem22}
\end{subequations}\vspace{-2mm}\\
By introducing a target minimum SINR $\gamma>0$, problem~\eqr{hbproblem22} can be equivalently reformulated as
\begin{subequations}
\begin{empheq}{align}
\max_{\mathbf{p}}\quad & \gamma\\
\textrm{s.t.} \quad & c_n p_n \ge \gamma, \forall n \in\mathcal{N}, \\\nonumber
& \text{(\ref{hbproblem}d)}, \text{and (\ref{hbproblem}e)}.
\end{empheq}
\label{hbproblem23}
\end{subequations}\vspace{-2mm}

From constraint~(\ref{hbproblem23}b), it follows that
\begin{equation}
p_n \ge \frac{\gamma}{c_n},\quad \forall n\in\mathcal{N}.
\label{plower}
\end{equation}
That is, a given SINR target $\gamma$ is feasible if and only if there exists a power allocation vector $\mathbf{p}$ satisfying (\ref{hbproblem}d), (\ref{hbproblem}e), and~\eqr{plower}. Since $\gamma>0$, constraint~(\ref{hbproblem}d) is always satisfied. Based on \eqr{plower}, the sum-power constraint~(\ref{hbproblem}e) results in
\begin{equation}
\sum_{n=1}^N \frac{\gamma}{c_n} \le 1,
\end{equation}
which can be equivalently expressed as
\begin{equation}
\gamma \le \frac{1}{\sum_{n=1}^N 1/c_n}.
\end{equation}
Therefore, the maximum feasible SINR is given by
\begin{equation}\label{optgamma}
\gamma^\star=\frac{1}{\sum_{n=1}^N 1/c_n}.
\end{equation}

According to Lemma~\ref{lemma1}, under the optimal power allocation $\mathbf{p}^\star$, all users achieve the same SINR. Hence, it follows that
\begin{equation}
p_n^\star=\frac{\gamma^\star}{c_n},\quad \forall n\in\mathcal{N}.
\end{equation}
Substituting \eqr{optgamma} into the above expression yields
\begin{equation}
p_n^\star=\frac{1/c_n}{\sum_{i=1}^N 1/c_i},\quad \forall n\in\mathcal{N},
\end{equation}
which completes the proof.
\end{IEEEproof}
\section{Simulation Results}
\begin{table}[t]
\centering
\caption{Simulation Parameters}\vspace{-1mm}
\label{parameter}
\begin{tabular}{lc}
\hline
\textbf{Parameter} & \textbf{Value} \\ \hline
Cable height ($d$) & $3$~m \\
Region size ($D_y$) & $20$~m \\
Relative permittivity ($\varepsilon_r$) & $1.26$ \\
Number of scatterers ($L$) & $10$ \\
Carrier frequency ($f_c$) & $3.5$~GHz\\
Scattering path gain ($\delta_\ell$) & $\delta_\ell \sim \mathcal{CN}(0,1)$ \\
Noise power ($\sigma^2$) & $-64$~dBm\\
Convergence tolerance ($\epsilon$) & $10^{-4}$\\ \hline
\end{tabular}\vspace{-4mm}
\end{table}

In this section, Monte Carlo simulations are conducted to evaluate both the system performance and the effectiveness of the proposed algorithms for the LCX based generalized pinching-antenna system under analog and hybrid beamforming. In the simulations, users are uniformly distributed in the $D_x\times D_y$ area, and the $L$ scatterers are randomly placed on the boundary (walls) with heights uniformly distributed in $[0,3]$~m. Three benchmarks are considered: i) a conventional fixed-antenna system without precoding, where the BS transmits the superposed signals of all users with equal power from $2K$ antennas deployed at the center of the service area at height $3$~m and spacing $\lambda/2$; ii) a conventional fixed-antenna MISO system with ZF beamforming, using the same antenna deployment; and iii) a single-port LCX feeding scheme in which transmission is performed solely through port~$\mathrm{L}$. For the LCX based schemes, three optimization settings are considered: ``Initial'' denotes the initial configuration defined in the proposed algorithms; ``Optimized w/o PA'' denotes slot activation optimization with equal power allocation; and ``Optimized w/ PA'' denotes slot activation and power allocation optimization. Unless otherwise specified, the simulation results are averaged over $10^4$ independent realizations of random user locations. The simulation parameters are listed in Table~\ref{parameter}.

\begin{figure}[!t]
\centering{
\subfigure[User 1 as the intended user]{\centering{\hspace{-4mm}\includegraphics[width=100mm]{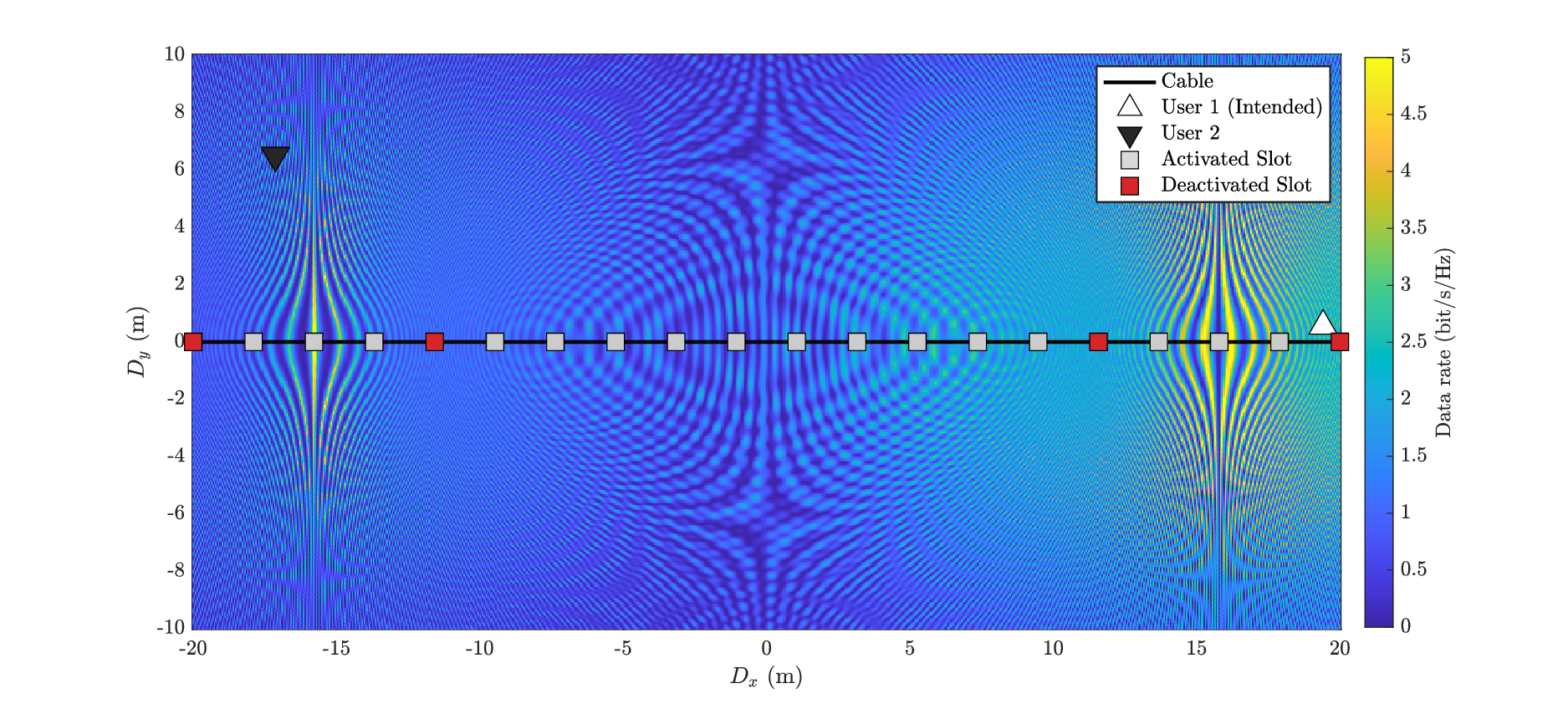}}}\vspace{-2mm}
\subfigure[User 2 as the intended user]{\centering{\hspace{-4mm}\includegraphics[width=100mm]{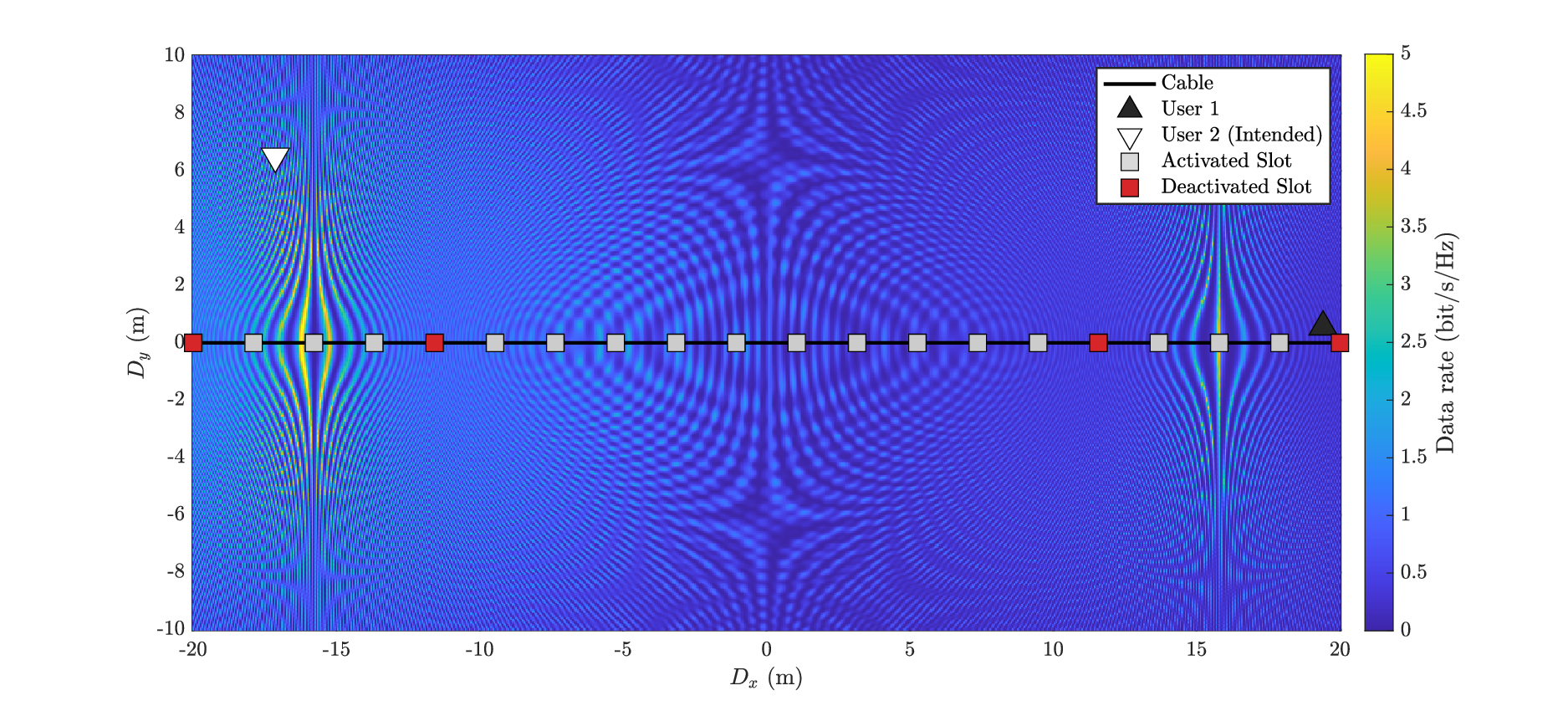}}}}\vspace{-2mm}
\caption{The achievable data rate distribution in the LCX based generalized pinching-antenna system with analog beamforming, where $\kappa=0.1$~dB/m, $D_x=40$~m, $K=1$, $N=2$, $M=20$, $L=10$, and $P_t=20$~dBm.}\vspace{-3mm}
\label{result01}
\end{figure}

\begin{figure}[!t]
\centering{
\subfigure[User 1 as the intended user]{\centering{\hspace{-4mm}\includegraphics[width=100mm]{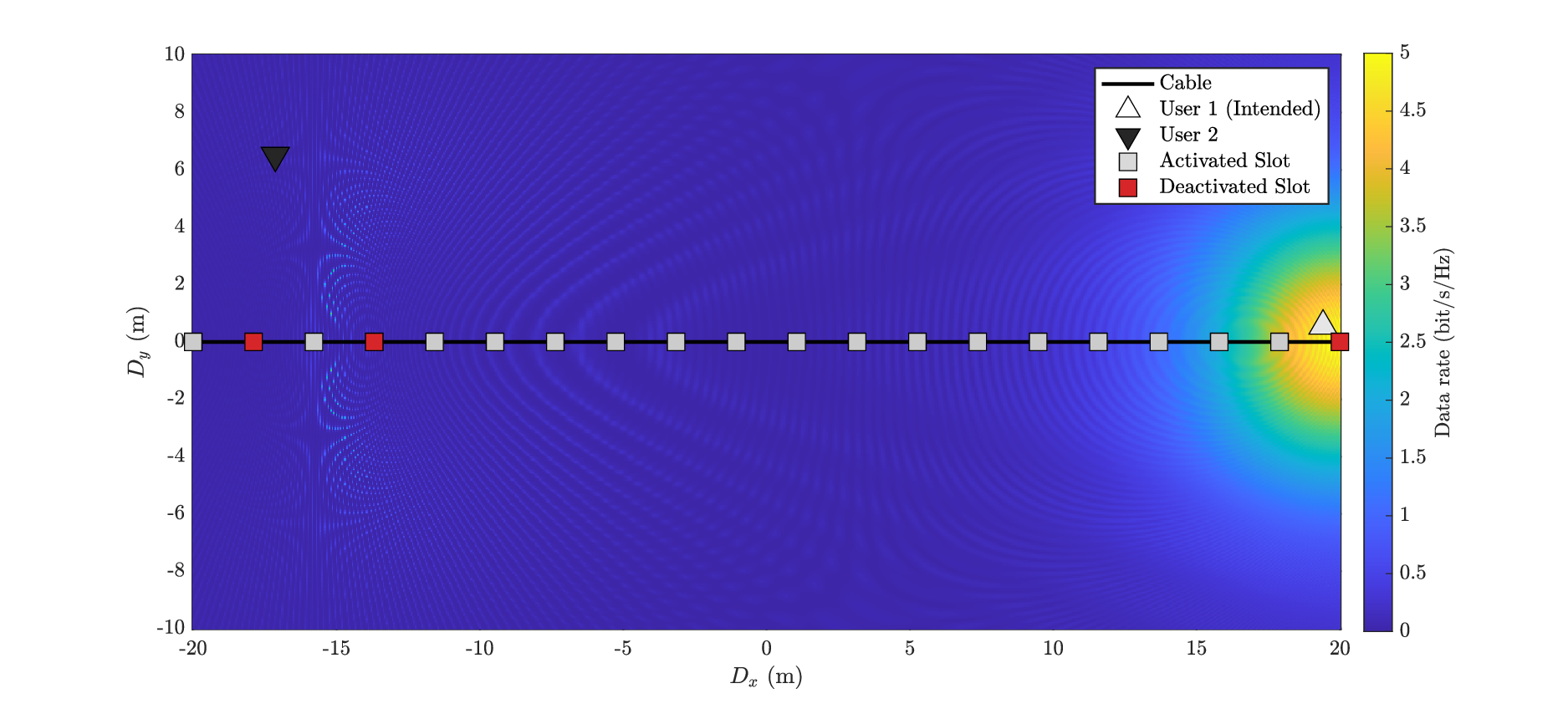}}}\vspace{-2mm}
\subfigure[User 2 as the intended user]{\centering{\hspace{-4mm}\includegraphics[width=100mm]{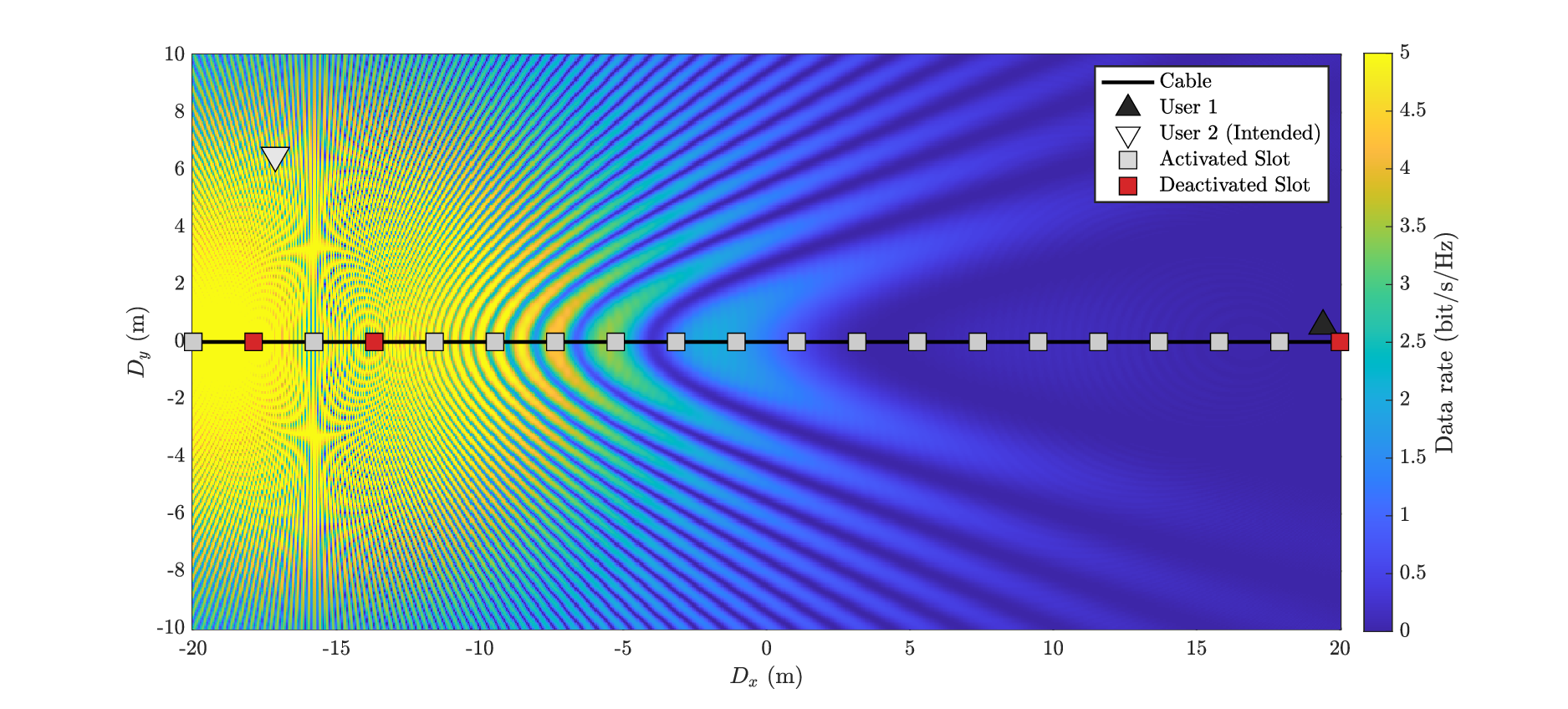}}}}\vspace{-2mm}
\caption{The achievable data rate distribution in the LCX based generalized pinching-antenna system with hybrid beamforming, where $\kappa=0.1$~dB/m, $D_x=40$~m, $K=1$, $N=2$, $M=20$, $L=10$, and $P_t=20$~dBm.}\vspace{-3mm}
\label{result02}
\end{figure}

Figs.~\ref{result01} and \ref{result02} illustrate the achievable data rate distributions obtained with the proposed solutions for analog and hybrid beamforming, respectively, based on the same realization of user locations. In the analog beamforming scheme shown in \fref{result01}, a relatively large number of radiation slots are activated, as pinching based beamforming relies on the selective activation of slots with different propagation induced phase shifts. This mechanism provides flexible phase control but also makes the performance highly sensitive to user locations, resulting in noticeable spatial variations in the data rate distribution. Moreover, the intrinsic attenuation of the cable leads to asymmetric signal strengths across the service region, thereby forming distinct coverage characteristics on the left and right sides of the cable. In contrast, the hybrid beamforming results in \fref{result02} exhibit a much smoother rate distribution due to the inclusion of ZF digital precoding, which effectively suppresses multi-user interference. As a result, when a user is located close to the cable, activating only the nearest slot is generally sufficient to ensure reliable transmission, whereas for users positioned farther away, multiple slots are activated to jointly enhance the received signal power through analog beamforming, compensating for the increased path loss.

\begin{figure}[!t]
\centering{
\subfigure[$K=1$ and $N=2$]{\centering{\includegraphics[width=84mm]{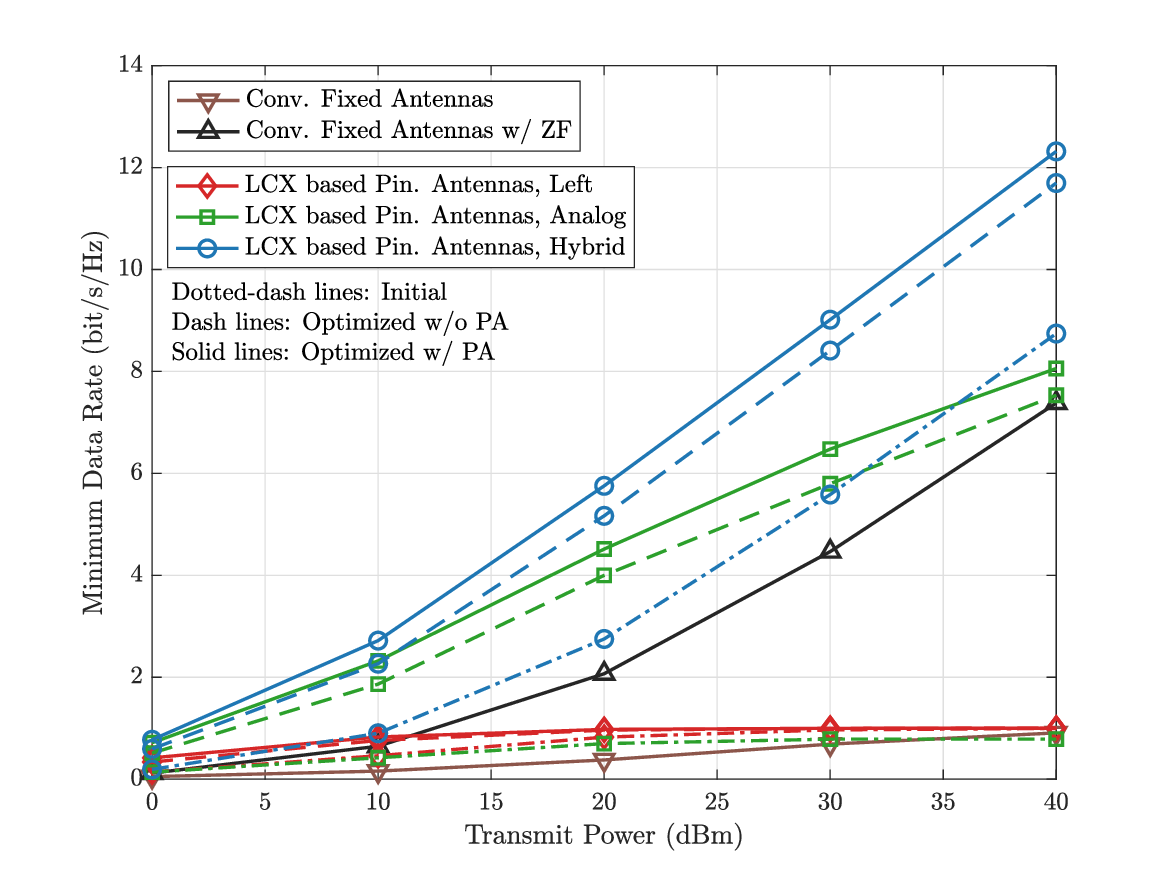}}}\vspace{-2mm}
\subfigure[$K=2$ and $N=4$]{\centering{\includegraphics[width=84mm]{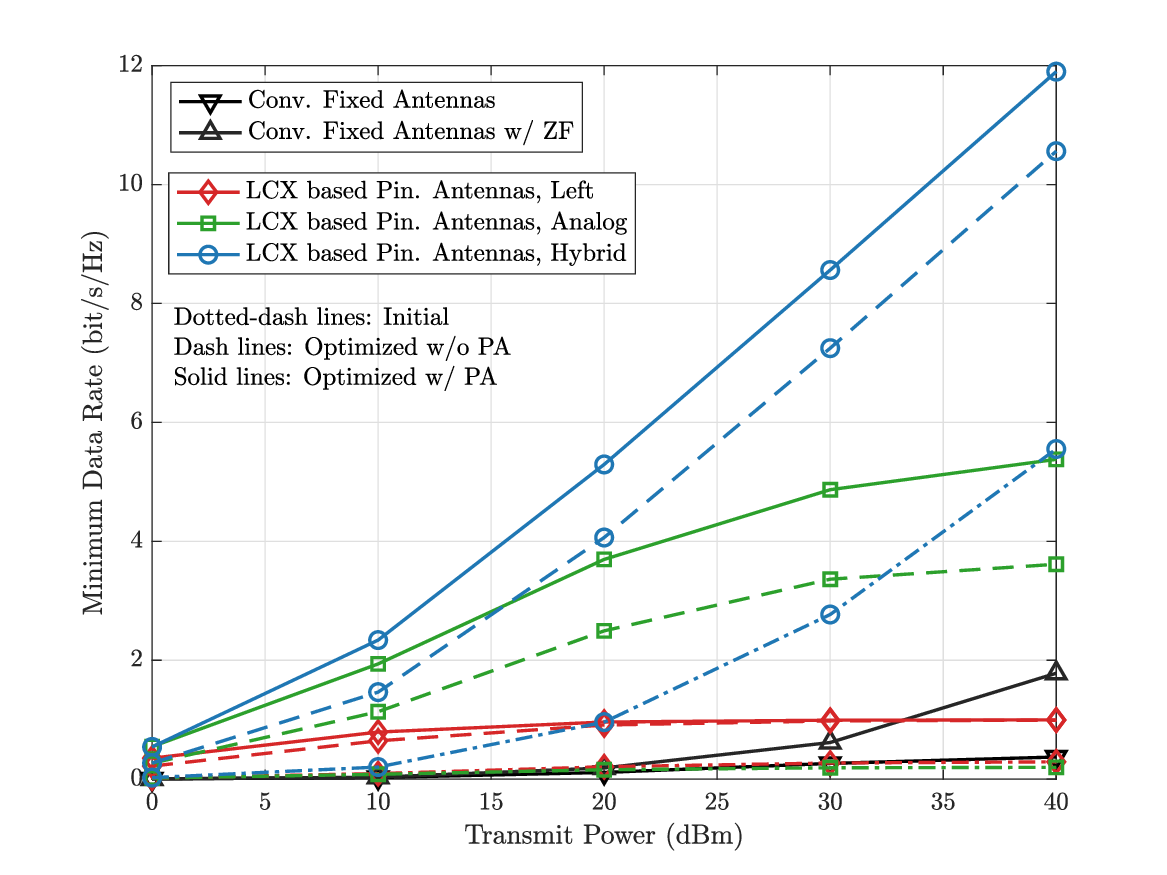}}}}\vspace{-2mm}
\caption{The impact of the transmit power, where $\kappa=0.1$~dB/m, $D_x=50$~m, and $M=50$.}\vspace{-4mm}
\label{result1}
\end{figure}

Figs.~\ref{result1}(a) and \ref{result1}(b) demonstrate the impact of the transmit power on the system performance under the single-cable and multi-cable scenarios, respectively. It can be observed that the proposed analog and hybrid beamforming schemes consistently outperform the benchmark methods, and the performance advantage becomes more significant in the multi-cable deployment. In the analog beamforming scheme, the performance is mainly determined by port selection and slot activation, which jointly determine the spatial distribution of the radiated signals. However, due to the presence of residual multi-user interference, the achievable data rate gradually saturates as the transmit power increases. In contrast, by suppressing inter-user interference through ZF digital precoding, the hybrid beamforming approach exhibits a markedly different trend. Specifically, when the transmit power exceeds $10$~dBm, the achievable data rate increases almost linearly with the transmit power, leading to a rapidly widening performance gap compared with analog beamforming, particularly in the multi-cable scenario with a larger number of served users.

\begin{figure}[!t]
\centering{\includegraphics[width=84mm]{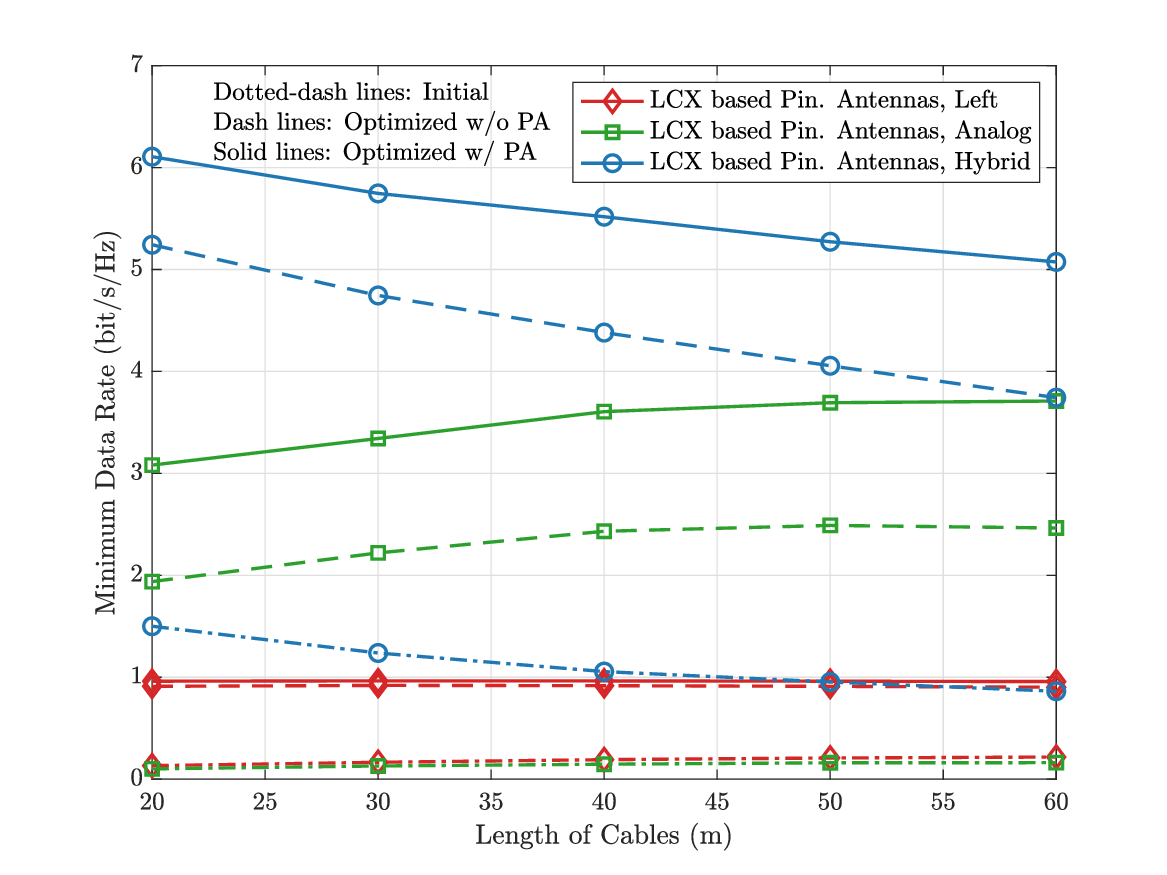}}
\caption{Impact of the length of cables, where  $\kappa=0.1$~dB/m, $K=2$, $N=4$, $\Delta d=1$~m, and $P_t=20$~dBm.}\vspace{-4mm}
\label{result2}
\end{figure}

\fref{result2} presents the impact of the cable length on the system performance, which also corresponds to an increase in the coverage region size. The slot spacing is fixed, and therefore increasing the cable length results in a larger number of radiation slots. It can be observed that the performance of the single-port feeding scheme remains nearly unchanged as the cable length increases, since both the desired and interfering components experience stronger intra-cable attenuation over longer propagation distances, which largely offsets the impact of additional distant slots on the resulting SINR. In contrast, the analog beamforming scheme benefits from the increased slot diversity, which provides additional degrees of freedom for pinching based beamforming and leads to improved performance. On the other hand, the hybrid beamforming scheme exhibits an opposite trend, as the accumulation of attenuation over distant slots degrades the channel conditions and reduces the efficiency of digital precoding. This observation is consistent with Remark~1 and further confirms the critical impact of propagation distance and attenuation in LCX based systems.

\begin{figure}[!t]
\centering{\includegraphics[width=84mm]{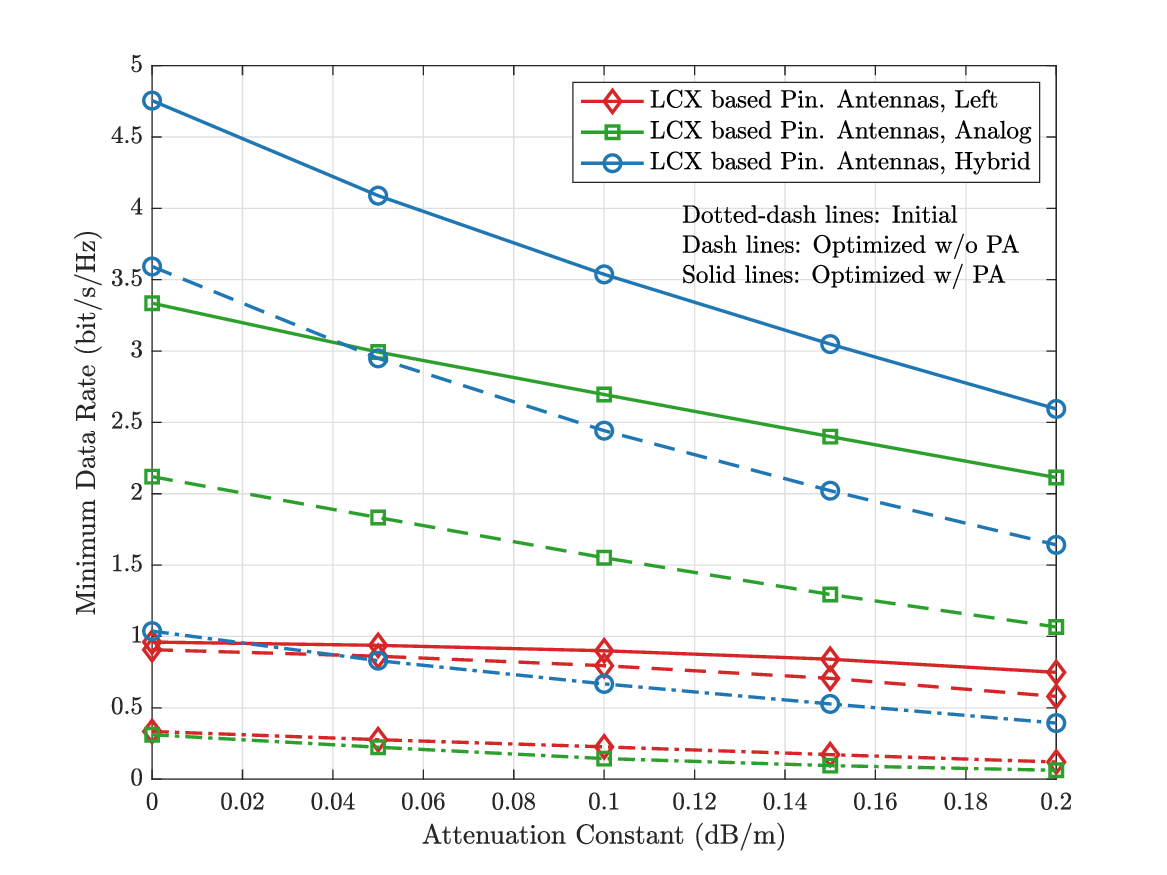}}
\caption{Impact of the attenuation constant, where $D_x=100$~m, $K=2$, $N=4$, $M=50$, and $P_t=20$~dBm.}\vspace{-4mm}
\label{result3}
\end{figure}

\fref{result3} illustrates the impact of the attenuation constant on the minimum achievable rate. As the attenuation constant increases, the performance of all considered schemes degrades due to stronger intra-cable power loss. It can be observed, however, that the proposed analog beamforming scheme exhibits a noticeably slower performance degradation. This is mainly attributed to port selection, where the desired signal of each user is transmitted through a single selected port, while interfering components typically propagate over longer cable distances and thus experience stronger attenuation. In contrast, the hybrid beamforming and benchmark schemes are more sensitive to attenuation, resulting in a steeper decline in the achievable minimum rate.

\begin{figure}[!t]
\centering{\includegraphics[width=84mm]{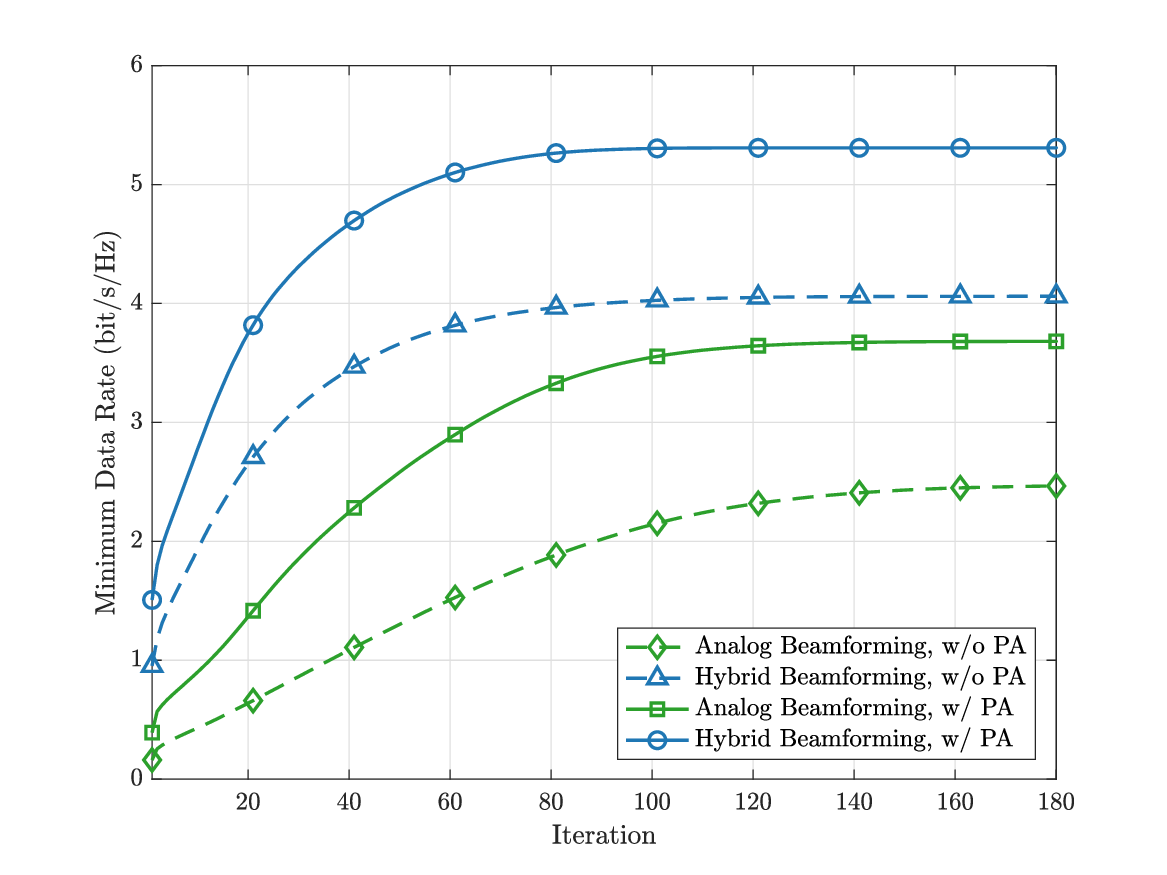}}
\caption{Convergence performance of the proposed algorithms, where $\kappa=0.1$~dB/m, $D_x=50$~m, $K=2$, $N=4$, $M=50$, and $P_t=20$~dBm.}\vspace{-4mm}
\label{result4}
\end{figure}
 
 \fref{result4} shows the convergence performance of Algorithm~\ref{cgalg} and Algorithm~\ref{cgalg2} for the analog and hybrid beamforming schemes, respectively. It can be observed that both algorithms converge to stable solutions within a limited number of iterations, which validates the convergence and complexity analysis discussed in the previous sections. The analog beamforming algorithm converges more slowly due to the additional port selection process, which enlarges the search space. Moreover, incorporating power allocation significantly accelerates convergence and leads to higher achievable minimum rates, demonstrating the effectiveness of joint beamforming and power allocation optimization.
 
 \begin{figure}[!t]
\centering{\includegraphics[width=84mm]{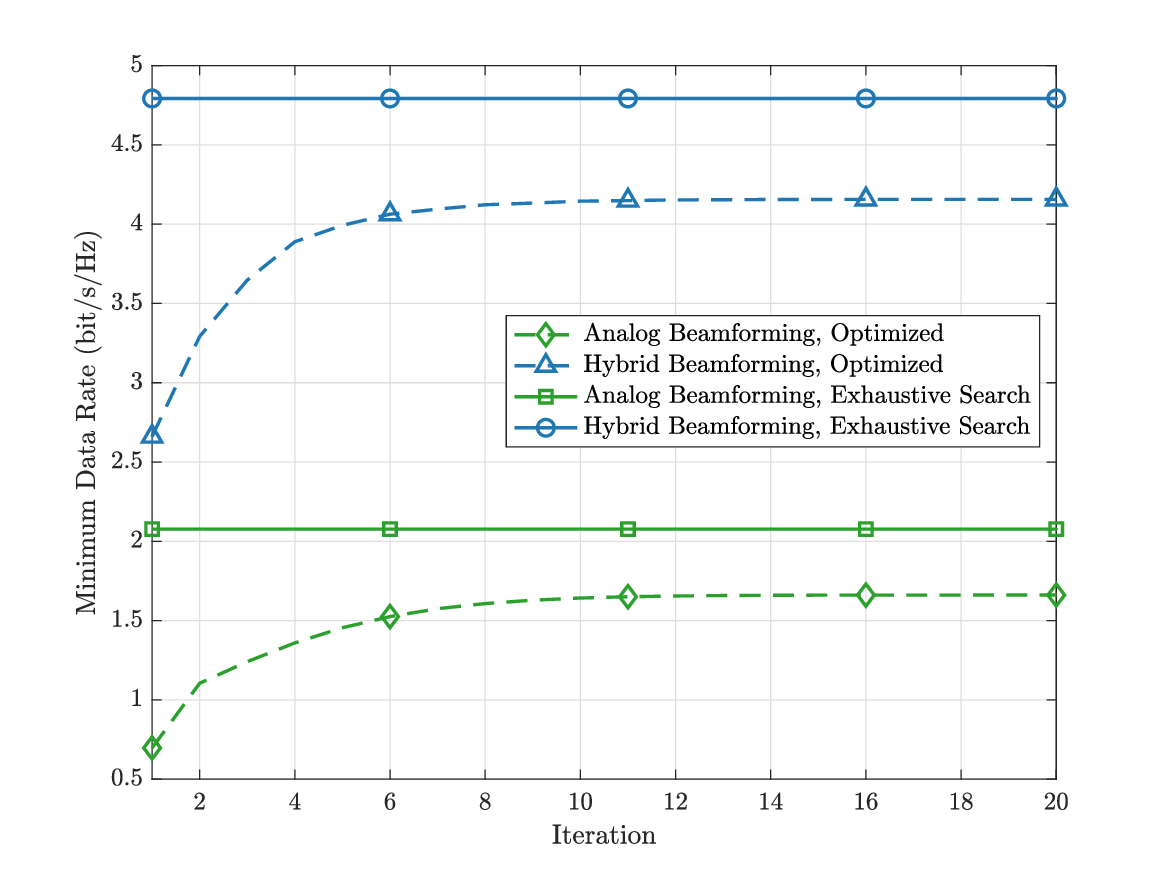}}
\caption{Comparison of the proposed algorithms with exhaustive search under equal power allocation, where $\kappa=0.1$~dB/m, $D_x=50$~m, $K=1$, $N=2$, $M=10$, and $P_t=20$~dBm.}\vspace{-4mm}
\label{result5}
\end{figure}

To further evaluate the optimality of the proposed algorithms, a small-scale system with fixed power allocation is considered, for which exhaustive search is included to obtain the global optimum. As shown in \fref{result5}, the proposed algorithms for the analog and hybrid beamforming schemes achieve approximately $80\%$ and $87\%$ of the performance achieved by the corresponding computationally intensive exhaustive-search algorithms, respectively. These results demonstrate that the proposed algorithms can be regarded as low-complexity suboptimal solutions.
\section{Conclusions}
This paper investigated an LCX based generalized pinching-antenna system with dual-port feeding. By enabling bidirectional signal injection along the cable, the proposed architecture increases the available spatial degrees of freedom, thereby facilitating more flexible and effective beamforming designs. Both analog and hybrid beamforming frameworks were considered under the proposed architecture, where the minimum data rate was maximized through the joint design of port selection, slot activation, and power allocation. The resulting optimization problems were efficiently addressed using matching based port selection, coalitional game based slot activation, and low-complexity power control methods. Simulation results confirm the advantages of dual-port feeding and further demonstrate the effectiveness of the proposed optimization methods under different transmit power levels and physical cable parameters. As a future research direction, practical channel estimation and robust transmission design under imperfect CSI constitute important extensions of the proposed framework.
\bibliographystyle{IEEEtran}
\bibliography{KaidisBib}
\end{document}